\documentclass[a4paper,reqno]{amsart}    
\usepackage[latin1]{inputenc}
\pdfoutput=1

\usepackage{amsmath}    
\usepackage{amsfonts}
\usepackage{amssymb}
\usepackage{amscd}
\usepackage{amsthm}     
\usepackage{graphicx}   
\usepackage{color}      
\usepackage{hyperref}   
\usepackage{enumitem}   
\usepackage{tikz}		
\usetikzlibrary{decorations.pathmorphing}
\usetikzlibrary{arrows.meta}

\usepackage[textsize=tiny]{todonotes}   

\allowdisplaybreaks[3]

\let\origfootnote\footnote
\renewcommand{\footnote}[1]{%
	\unskip%
	\linespread{1}%
	\origfootnote{#1}%
	\linespread{1.2}%
}

\let\emph\relax 
\DeclareTextFontCommand{\emph}{\bfseries}

\def\bC{{\mathbb C}}       
\def\bM{{\mathbb M}}  
\def\bR{{\mathbb R}} 
\def\bN{{\mathbb N}}  

\def\bH{{\mathbb H}}
\def\mcd{{\mathcal{D}}}
\def\dd{{\mathrm{d}}}
\def\AdS{\mathrm{AdS}}

\def\PAdS{\mathrm{PAdS}}

\newtheorem{theorem}{Theorem}[section]
\newtheorem{lemma}[theorem]{Lemma}
\newtheorem{proposition}[theorem]{Proposition}
\newtheorem{corollary}[theorem]{Corollary}
\theoremstyle{definition}
\newtheorem{definition}[theorem]{Definition}

\newtheorem{remark}[theorem]{Remark}

\numberwithin{equation}{section}
\numberwithin{theorem}{section}

\begin{document}

 	
 	\author{Claudio Dappiaggi}
 	\address{Dipartimento di Fisica, Universit\`a degli Studi di Pavia, Via Bassi, 6, I-27100 Pavia, Italy \\
 		Istituto Nazionale di Fisica Nucleare -- Sezione di Pavia, Via Bassi, 6, I-27100 Pavia, Italy}
 	\email{claudio.dappiaggi@unipv.it}

 	\author{Hugo R. C. Ferreira}
 	\address{Istituto Nazionale di Fisica Nucleare -- Sezione di Pavia, Via Bassi, 6, I-27100 Pavia, Italy}
 	\email{hugo.ferreira@pv.infn.it}

 	
 	\title[Hadamard states for a scalar field in AdS]{On the algebraic quantization of a massive scalar field in anti-de-Sitter spacetime.}

 	
 	\date{\today}

 	
 	\begin{abstract}
 	We discuss the algebraic quantization of a real, massive scalar field in the Poincar\'e patch of the $(d+1)$-dimensional anti-de Sitter spacetime, with arbitrary boundary conditions. By using the functional formalism, we show that it is always possible to associate to such system an algebra of observables enjoying the standard properties of causality, time-slice axiom and F-locality. In addition, we characterize the wavefront set of the ground state associated to the system under investigation. As a consequence, we are able to generalize the definition of Hadamard states and construct a global algebra of Wick polynomials.
 	\end{abstract}

 	

 	\maketitle


\section{Introduction}

Algebraic quantum field theory is a mathematically rigorous, axiomatic approach  which has evolved vastly since its formulation due to Haag and Kastler \cite{Haag:1963dh}. Striking successes have come from its application to the quantization of field theories in curved backgrounds. In this context, most notable has been the characterization of the class of physically acceptable states in terms of the singular structure of the truncated two-point function, the construction of the algebra of Wick polynomials and the development of perturbation theory, including a full-fledged analysis of renormalization and of the related freedoms. Several reviews on these topics have been written, e.g.~\cite{Brunetti:2015vmh,Benini:2013fia,Hollands:2014eia,Rejzner:2016hdj}.

All these works rely on rather mild assumptions and they encompass almost all scenarios of physical interest. In particular, at a geometric level, the only requirement, which is always made, is to consider a globally hyperbolic spacetime of arbitrary dimension. In this way, on the one hand, no pathology, such as closed timelike curves, can occur at the level of causal structure, whereas, on the other hand, it is possible to analyze the dynamics of most field theories in terms of an initial value problem. 

Although the necessity and the effectiveness of such assumption is at first glance unquestionable, it is not hard to realize that we are discarding many scenarios of physical interest whose underlying model does not fit in this scheme. An example, which has been recently investigated by one of us, is the Casimir effect \cite{Dappiaggi:2014gea}. Another important example is the case of field theories in anti-de Sitter (AdS) spacetime. This is a maximally symmetric solution of the Einstein equations with negative cosmological constant which is the arena for the renown AdS/CFT correspondence --- see for example the recent monograph \cite{Ammon:2015wua} ---, which has also been investigated rather extensively from the viewpoint of algebraic quantum field theory \cite{Duetsch:2002hc,Duetsch:2002wy,Rehren:2000tp,Ribeiro:2007hv}. 

Although the AdS/CFT correspondence offers a clear reason for the relevance of studying quantum field theories in AdS spacetime, we feel that, prior to addressing any issue related to this very important conjecture, it is worth filling another open gap. As a matter of fact, up to now there is no full-fledged and systematic study of the quantization in the algebraic approach of free fields on AdS. Several works which address specific points are available in the literature, notable examples being the seminal paper \cite{Avis:1977yn} and \cite{Allen:1985wd}. Yet several other questions remain open and algebraic quantum field theory seems to offer to best framework the answer them. 

In this paper we will be working in $\PAdS_{d+1}$, the so-called Poincar\'e patch of a $(d+1)$-dimensional AdS spacetime, which is the usual framework for the AdS/CFT correspondence, and we will be considering a real, massive scalar field. Via a conformal rescaling we will reformulate the problem in $\mathring{\bH}^{d+1}$, which is nothing but Minkowski spacetime with one Euclidean space coordinate, say $z$, constrained to have domain in the half line $(0,\infty)$. The dynamics becomes that of the wave equation with, in addition, a potential $\frac{m^2}{z^2}$ singular at $z=0$, $m$ being the mass of the field.

Since the underlying background is not globally hyperbolic, the dynamics cannot be solved in terms of an initial value problem, and a boundary condition at $z=0$ has to be imposed. This opens the question of which is the class of admissible boundary conditions of Robin type and on the correct mathematical method to implement them \footnote{A more general approach, known as Wentzell boundary conditions, could be considered. A preliminary investigation in this direction is available in \cite{Zahn:2015due}.}. Since the potential and consequently the generic solution to the equation of motion are expected to be singular at $z=0$, the standard paradigm of consider a linear combination of the field and of its normal derivative to the boundary is not applicable. A first investigation of this problem can be found in \cite{Ishibashi:2004wx}. Further elaborating this avenue of research, in \cite{Dappiaggi:2016fwc}, using the theory of Sturm-Liouville operators, the class of all possible ground states compatible with the boundary conditions were constructed. More precisely we built positive bidistributions in $\mathring{\bH}^{d+1}$ which are weak solutions of the equations of motion, implementing a chosen, admissible boundary condition and moreover invariant under the action of the isometry group of the background. 

Here, we start from these results, further elaborating on the algebraic quantization of a real, massive scalar field in $\PAdS_{d+1}$ by means of its equivalent realization in $\mathring{\bH}^{d+1}$. The first point that we address is the construction of all causal propagators, one for each admissible boundary condition, realizing them as the antisymmetric part of the states built in \cite{Dappiaggi:2016fwc}. Since the underlying spacetime is not globally hyperbolic, we cannot rely a priori on any of the standard properties of the propagators \cite{BGP}, having to prove them anew. Particularly interesting is the singular structure, encoded in the wavefront set, since, in addition to the same contribution as for wavelike operators on globally hyperbolic spacetimes, it includes a novel bit, namely, in the singular support, there are also those pairs of points which can be connected by a null geodesic reflected at the boundary $z=0$. 

It is important to stress that our analysis relies strongly on the relatively simple form both of the underlying metric and of the equation of motion, which allows for writing explicit formulae. Yet, especially when dealing with the study of the wavefront set of the distributions involved or when thinking about a generalization to a wider class of spacetimes, it is more effective to resort to more specific techniques and tools, such as those going under the name of b-calculus. A very interesting work in this direction has recently appeared \cite{Wrochna:2016ruq}, though only one boundary condition (that of Dirichlet type) has been considered. We remark that, though with different methods, our findings are in agreement with those of \cite{Wrochna:2016ruq}.  It is also worth mentioning that a special case of our model is the one with vanishing mass, which corresponds in $\mathring{\bH}^{d+1}$ to the so-called Casimir-Polder system, investigated for Dirichlet boundary conditions in \cite{Dappiaggi:2014gea}.

The identification and the analysis of the structural properties of all possible causal propagators can be seen as the starting point of the algebraic quantization procedure. As a matter of fact, using the functional formalism we construct the algebra of observables for the system under investigation, proving not only causality but also the time-slice axiom and the so-called {\em F-locality}. While the first stems from the structural properties of the causal propagator, the second is often ascribed to the existence of a Cauchy problem. We show that, by a suitable extension of the observables beyond those smooth and compactly supported, it is possible to make sense of the time-slice axiom also on a non globally hyperbolic spacetime, regardless of the choice of boundary condition. The third property was first proposed in \cite{Kay:1992es} and it dictates that the restriction of the algebra of observables to any globally hyperbolic subregion of the underlying spacetime should be $*$-isomorphic to the one constructed intrinsically in such region using the usual procedure. 

As the last point, we focus once more on the ground states built in \cite{Dappiaggi:2016fwc}, realizing them as full-fledged states on the algebra of observables. The first step of our analysis consists of characterizing their wavefront set, showing on the one hand that it has a richer singular structure, since the singular support includes pair of points connected by null geodesics reflected at the boundary, while on the other hand that its restriction to any globally hyperbolic subregion coincides with the wavefront set of Hadamard states as per \cite{Radzikowski:1996pa,Radzikowski:1996ei}. 

This result prompts us to call of Hadamard form all those states in $\PAdS_{d+1}$ whose associated truncated two-point function has the same wavefront set of the ground state or, equivalently, has an integral kernel differing from the one of the ground state by a smooth term. To strengthen our definition, we show how to construct the extended algebra of Wick polynomials. 

The synopsis of the paper is the following. In Section~\ref{sec:AdS} we review the geometric setting, in particular the Poincar\'e domain of an anti-de Sitter spacetime as well as the associated notion of chordal distance. In Section~\ref{sec:KG} we consider a real, massive scalar field in $\PAdS_{d+1}$ and we construct the associated equation in $\mathring{\bH}^{d+1}$. Eventually we review the choice of possible boundary conditions and their implementation in the language of a Sturm-Liouville problem. Section~\ref{sec:AQFT} is the core of the paper. In subsection~\ref{sec:causalpropagator}, we review the construction of the possible causal propagator, proving their structural properties and we analyze the associated wavefront set. These results are used in subsections~\ref{sec:offshellalgebra} and \ref{sec:onshellalgebra} to construct respectively the off-shell and the on-shell algebras of observables, proving in addition their structural properties, in particular the time-slice axiom and F-locality. Subsequently, in subsections~\ref{sec:hadamardstates} and \ref{sec:wickordering} we focus on the ground states of the theory under investigation, studying their wavefront set and formulating accordingly an extended notion of Hadamard states. Finally, we show how to construct a global algebra of Wick polynomials. Open problems are discussed succinctly in the conclusions.


\section{Anti-de Sitter and the Poincar\'e domain}
\label{sec:AdS}

Anti-de Sitter spacetime, $\AdS_{d+1}$ ($d\geq 2$), is the maximally symmetric solution of the $(d+1)$-dimensional Einstein's equation  with a negative cosmological constant $\Lambda$. It is defined in the embedding space $\bR^{d+2}$ endowed with line element
$$\dd s^2 = -\dd X^2_0 - \dd X^2_1 + \sum\limits_{i=2}^{d+1} \dd X^2_i \, , $$ 
where $(X_0,...,X_{d+1})$ are the standard Cartesian coordinates, as the region
\begin{equation} \label{eq:covering_space}
	-X^2_0-X^2_1+\sum_{i=2}^{d+1}X^2_i=-\ell^2 \, , \qquad \ell^2 \doteq -\frac{d(d-1)}{\Lambda} \, .
\end{equation}

The \emph{Poincar\'e fundamental domain}, $\PAdS_{d+1}$, is identified via the coordinate transformation
\begin{equation} \label{eq:Poincare_chart}
	\arraycolsep=1.4pt\def\arraystretch{2}  
	\left\{\begin{array}{l}
		X_0 = \dfrac{\ell}{z}t \, , \\
		X_i = \dfrac{\ell}{z} x_i \, , \quad i=1,...,d-1,\\
		X_d=\ell\left(\dfrac{1-z^2}{2z}+\dfrac{-t^2+\delta^{ij}x_i\,x_j}{2z}\right) \, , \\
		X_{d+1}=\ell\left(\dfrac{1+z^2}{2z}-\dfrac{-t^2+\delta^{ij}x_i\,x_j}{2z}\right) \, ,
	\end{array}\right.
\end{equation}
where $t, \, x_i \in \bR$ and $z\in \bR_{>0}$. This translates the constraint which descends from the identity $X_d+X_{d+1}=\frac{\ell}{z}$, hence  $\PAdS_{d+1}$ covers only half of the full $\AdS_{d+1}$ (see Fig.~\ref{fig:AdS}). In addition, the line element of the Poincar\'e domain becomes
\begin{equation} \label{eq:Poincare_metric}
	\dd s^2 = \frac{\ell^2}{z^2} \left(-\dd t^2+\dd z^2+\delta^{ij} \dd x_i \dd x_j\right) \, , \qquad i,j=1,...,d-1 \, ,
\end{equation}
where $\delta^{ij}$ is the Kronecker delta.
%
%
\begin{figure}
	\centering
	\includegraphics[scale=0.4]{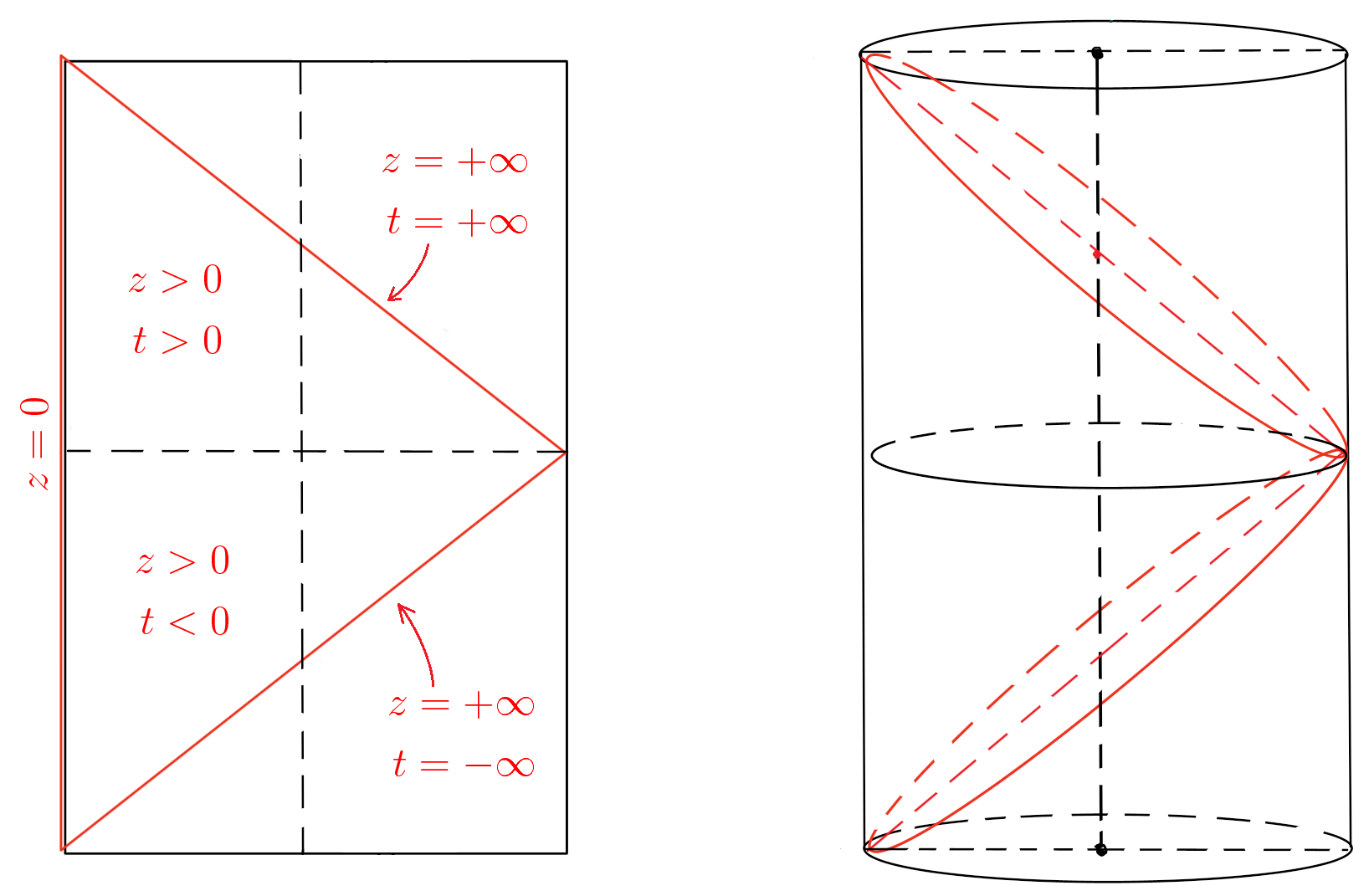}
	\caption{\label{fig:AdS}Conformal diagram of $\AdS_{d+1}$ and of the Poincar\'e domain $\PAdS_{d+1}$ with $(d-1)$ and $(d-2)$ spatial dimensions removed. 
		\vspace*{1ex}
	}
\end{figure}

Observe that $\PAdS_{d+1}$ is conformal to a portion of Minkowski spacetime, the ``upper-half plane''
$$\mathring{\bH}^{d+1}\doteq\{(t,x_1, \ldots, x_{d-1},z)\in\bR^{d+1}\;|\;z>0\} \, , $$ 
where the same Cartesian coordinates were adopted as in \eqref{eq:Poincare_metric}. Endowing $\mathring{\bH}^{d+1}$ with the standard Minkowskian metric $\eta$, then $\eta=\Omega^2 g=\frac{z^2}{\ell^2}g$ where $g$ is the metric \eqref{eq:Poincare_metric} of $\PAdS_{d+1}$ and $\Omega=\frac{z}{\ell}$ is a conformal factor. For later convenience, we observe that \eqref{eq:Poincare_metric} is meaningful also for negative values of $z$, the only singularity occurring at $z=0$. Therefore we can introduce the map which reflects points along the $z=0$ hyperplane:
\begin{equation}\label{inversion}
\iota_z:\bR^{d+1}\to\bR^{d+1},\quad (\underline{x},z)\mapsto\iota_z(\underline{x},z)= (\underline{x},-z) \, ,
\end{equation} 
where $\underline{x}=(t,x_1,...,x_{d-1})$.

An important concept, namely when discussing maximally symmetric quantum states, is that of invariant distance in AdS. Intrinsically, one can define the geodesic distance $s$ on $\PAdS_{d+1}$ between two arbitrary points $x$ and $x'$ and the Synge's world function $\sigma$ given by $\sigma(x,x') \doteq \frac{1}{2}s(x,x')^2$. Alternatively, one can define the chordal distance $s_{\rm e}$ between $x$ and $x'$ through the embedding space $\bR^{d+2}$ and the Synge's world function defined on $\bR^{d+2}$ as
\begin{equation} \label{eq:sigmae}
	\sigma_{\rm e}(x,x') \doteq \frac{1}{2}s_{\rm e}(x,x')^2 = \frac{1}{2} \eta^{AB} (X_A - X'_A) (X_B - X'_B) \, ,
\end{equation}
with $x$ and $x^\prime$ constrained by \eqref{eq:covering_space}. The intrinsic and chordal distances defined in this way are related by 
\begin{equation} \label{eq:relationsigmas}
	\cosh \left(\frac{s}{\ell}\right) = 1 + \frac{s_{\rm e}^2}{2 \ell^2} \, , \qquad
	\cosh \left(\frac{\sqrt{2\sigma}}{\ell}\right) = 1 + \frac{\sigma_{\rm e}}{\ell^2}
\end{equation}
(see e.g.~Section 2.4 of \cite{Kent:2013}). Finally, in the Poincar\'e fundamental domain, the Synge's world function $\sigma$ in $\PAdS_{d+1}$ and the Synge's world function $\sigma_{\bM}$ in $\mathring{\bH}^{d+1}$ are related by
\begin{equation} \label{eq:sigmasigmaM}
	\cosh^2 \left(\frac{\sqrt{2\sigma}}{\ell}\right) = 1 + \frac{\sigma_{\bM}}{2zz'} = \frac{\sigma_{\bM}^{(-)}}{2zz'} \, ,
\end{equation}
where, with $i,j=1,\ldots,d-1$,
\begin{align*}
	\sigma_{\bM} \doteq \frac{1}{2} \left[ - (t-t')^2 + \delta^{ij} (x_i - x'_i)(x_j - x'_j) + (z-z')^2 \right] \, ,
\end{align*}
and $\sigma_{\bM}^{(-)} \doteq \iota_z(\sigma_{\bM})$.

\vspace*{0.5ex}

\begin{remark}
	In the rest of the paper, we set $\ell \equiv 1$.
\end{remark}

\vspace*{1.5ex}


\section{Massive scalar field in AdS}
\label{sec:KG}


\subsection{Klein-Gordon equation}
\label{sec:KG_equation}

We consider a real, massive scalar field in the Poincar\'e domain $\phi: \PAdS_{d+1}\to\bR$ such that 
\begin{equation}\label{eq:dynamicsP}
	P\phi=\left(\Box_g - m_0^2 - \xi R \right)\phi=0 \, ,
\end{equation}
where $\Box_g$ is the D'Alembert wave operator built out of the metric \eqref{eq:Poincare_metric}, $m_0$ is the mass of the scalar field, $\xi$ is the scalar-curvature coupling constant and $R = - d(d+1)$ is the Ricci scalar. 

In order to study the solutions of this equation, we recall that $\PAdS_{d+1}$ is conformal to $\mathring{\bH}^{d+1}$ and translate \eqref{eq:dynamicsP} into a partial differential equation intrinsically defined in $\mathring{\bH}^{d+1}$. This is a standard procedure, see e.g.~Appendix D of \cite{Wald}. Let $\phi:\PAdS_{d+1}\to\bR$ be any solution of \eqref{eq:dynamicsP} and let $\Phi\doteq\Omega^{\frac{1-d}{2}}\phi$. The latter can be read as a scalar field $\Phi:\mathring{\bH}^{d+1}\to\bR$, solution of the equation
\begin{equation} \label{eq:conformally_rescaled_dynamics}
	P_\eta \Phi = \left(\Box_\eta-\frac{m^2}{z^2}\right) \Phi=0 \, ,
\end{equation}
in which $\Box_\eta$ is the standard wave operator built out of the Minkowski metric $\eta$ and
\footnote{Note that $m^2$ differs from the ``effective mass'' $m_0^2+\xi R$ used in other references.}
$m^2 \doteq m_0^2+(\xi-\frac{d-1}{4d})R$. In other words, the Klein-Gordon equation in $\PAdS_{d+1}$ is transformed to a wave equation on $\mathring{\bH}^{d+1}$ with a potential, singular at $z=0$. 

The first step in the algebraic quantization scheme consists of characterizing the set of smooth solutions of \eqref{eq:conformally_rescaled_dynamics}. To this end, it is necessary to construct the retarded-minus-advanced fundamental solution $G_{\bH} \in \mathcal{D}'(\mathring{\bH}^{d+1} \times \mathring{\bH}^{d+1})$ associated to $P_\eta$, also known as the causal propagator. Although its principal symbol is normally hyperbolic and we can consider $\mathring{\bH}^{d+1}$ as embedded in Minkowski spacetime, the presence of the singular potential $\frac{m^2}{z^2}$ does not allow us to use the standard theorems of existence and uniqueness for the advanced and retarded Green operators \cite{BGP}.

Hence, one needs to resort to a detailed analysis of the case in hand to construct explicitly the causal propagator. This is a lengthy calculation discussed thoroughly in \cite{Dappiaggi:2016fwc}. In the following we report succinctly the main tools and results of this paper.


\subsection{Boundary conditions}

Following \cite{Dappiaggi:2016fwc}, in view of the invariance of the metric \eqref{eq:Poincare_metric} under translations along the directions orthogonal to $z$, we take the Fourier transform,
\begin{equation} \label{eq:Fouriertransf}
\Phi(\underline{x},z) = \int_{\bR^d} \frac{\dd^d\underline{k}}{(2\pi)^{\frac{d}{2}}} \, e^{i\underline{k}\cdot \underline{x}} \, \widehat{\Phi}_{\underline{k}}(z) \, ,
\end{equation}
with $\underline{x} \doteq (t, x_1, \ldots, x_{d-1})$, $\underline{k} \doteq (\omega, k_1, \ldots, k_{d-1})$. The modes $\widehat{\Phi}_{\underline{k}}$ are solutions of
\begin{equation} \label{eq:STeq}
L \, \widehat{\Phi}_{\underline{k}} = \lambda \, \widehat{\Phi}_{\underline{k}}(z) \, , \qquad
\lambda \equiv q^2 \doteq \omega^2 - \displaystyle\sum_{i=1}^{d-1} k_i^2 \, ,
\end{equation}
where we introduce the singular Sturm-Liouville operator 
$$ L \doteq -\frac{\dd^2}{\dd z^2}+\frac{m^2}{z^2} \, , \qquad z\in(0,\infty) \, , $$
and $m^2$ is the parameter in \eqref{eq:conformally_rescaled_dynamics}. Following \cite[\S 10]{Zettl:2005}, we introduce the following definition.

\begin{definition}
	For any $z_0\in (0,\infty)$ we call \emph{maximal domain} associated to $L$
	\begin{equation*}
	D_{\rm max}(L;z_0) \doteq \left\{\Psi:(0,z_0)\to\bC \;|\; \Psi, \tfrac{\dd\Psi}{\dd z} \in AC_{\rm loc}(0,z_0) \; \textrm{and} \; \Psi,L(\Psi)\in L^2(0,z_0) \right\},
	\end{equation*}
	where $AC_{\rm loc}(0,z_0)$ stands for the collection of all complex-valued, locally absolutely continuous functions on $(0,z_0)$.
\end{definition}

One of the key advantages of the notion of maximal domain is the possibility of using it as the starting point to implement boundary conditions at $z=0$. Still following \cite{Zettl:2005}, we introduce the fundamental pair of solutions of $L \Phi = q^2 \Phi$ as 
\begin{subequations}
\begin{align}
\Phi_1(z) &= \sqrt{\dfrac{\pi}{2}} q^{-\nu}\sqrt{z} \, J_\nu(qz) \, ,\label{sol1} \\
\Phi_2(z) &= \begin{cases}
- \sqrt{\dfrac{\pi}{2}} \, q^{\nu} \sqrt{z}  \, J_{-\nu}(qz) \, , & \nu \in (0,1) \, , \\
- \sqrt{\dfrac{\pi}{2}} \sqrt{z} \left[ Y_{0}(qz) - \dfrac{2}{\pi} \log(q) \right] \, , & \nu = 0 \, ,
\end{cases}\label{sol2}
\end{align}
\end{subequations}
where $\nu \doteq \frac{1}{2}\sqrt{1+4m^2}$, while $J_{\nu}$ and $Y_{\nu}$ stand for the Bessel functions of first and second kind, respectively. Since, per definition,  fundamental solutions ought to be square-integrable in $(0,\infty)$, only $\Phi_1$ exists for $\nu\geq 1$.

\begin{definition} \label{def:boundary_condition}
	We say that $\Psi_{\alpha} : (0,\infty) \to \bC$ satisfies an \emph{$\alpha$-boundary condition} (compatible with $L$) at the endpoint 0, or equivalently that $\Psi_{\alpha} \in D_{\rm max}(L;z_0)$, if the following two conditions hold true:
	\begin{itemize}
		\item[(i)] there exists $z_0\in(0,\infty)$ such that $\Psi_{\alpha} \in D_{\rm max}(L;z_0)$
		\item[(ii)] there exists $\alpha\in (0,\pi]$ such that
		$$\lim_{z\to 0} \big\{ {\cos(\alpha) \, W_z[\Psi_{\alpha},\Phi_1]+\sin(\alpha) \, W_z[\Psi_{\alpha},\Phi_2]} \big\} = 0 \, , $$
		where $W_z[\Psi_{\alpha},\Phi_i] = \Psi_{\alpha}\frac{\dd \Phi_i}{\dd z} - \Phi_i\frac{\dd\Psi_{\alpha}}{\dd z}$, $i=1,2$, is the Wronskian.
	\end{itemize}
\end{definition}

\begin{remark}
	The $\alpha$-boundary condition is more commonly known as {\em Robin} boundary condition. In particular, if $\alpha=\pi$ we refer to it as being a {\em Dirichlet} or {\em Friedrichs} boundary condition, while $\alpha=\frac{\pi}{2}$ is sometimes referred to as {\em Neumann} boundary condition. More details may be found in \cite{Dappiaggi:2016fwc}. Notice that, up to an irrelevant sign, we could have replaced $\alpha=\pi$ with $\alpha=0$ and these two values can be interchanged without altering any result. Our choice is dictated by later notational convenience when dealing with the case $\nu\in(0,1)$.
\end{remark}

\begin{remark}
	In the following sections many quantities, such as the field solutions and the Green operators, depend explicitly not only on the boundary condition, parametrized by $\alpha$, but also on $\nu$ and the spacetime dimension $d+1$. For notation simplicity, we only make the $\alpha$ dependency explicit.
\end{remark}


\section{Algebraic quantum field theory in AdS}
\label{sec:AQFT}


\subsection{Causal propagator}
\label{sec:causalpropagator}

Following the algebraic scheme of quantization, see for example \cite[\S 3]{Brunetti:2015vmh} or \cite{Benini:2013fia}, the building block necessary to construct both the space of smooth solutions of \eqref{eq:dynamicsP} and to realize in a covariant way the canonical commutation relations is the \emph{causal propagator} or \emph{retarded-minus-advanced fundamental solution}, $G_{\alpha} \in \mcd^\prime(\PAdS_{d+1}\times\PAdS_{d+1})$. Here, the $\alpha$ subscript stands for the $\alpha$-boundary condition that is applied at $z=0$, c.f.~Definition~\ref{def:boundary_condition}. From a structural point of view we realize it as $G_{\alpha} = G_{\alpha}^+-G_{\alpha}^-$, the difference between the retarded $(+)$ and the advanced $(-)$ fundamental solutions of $P$. In other words, we look for $G_{\alpha}^\pm \in \mcd^\prime(\PAdS_{d+1}\times\PAdS_{d+1})$ such that 
\begin{equation} \label{advanced_minus_retarded}
P\circ G_{\alpha}^\pm = G_{\alpha}^\pm\circ P = \mathbb{I} \, , \quad 
\textrm{supp}(G^\pm_\alpha(f)) \subseteq J^\pm(\textrm{supp}(f)) \, ,
\end{equation}
where $\mathbb{I}$ is the identity on $C^\infty_0(\PAdS_{d+1})$, of which $f$ is an arbitrary element. The symbols $J^\pm$ indicate the causal future $(+)$ and the causal past $(-)$ in $\PAdS_{d+1}$. Observe that \eqref{advanced_minus_retarded} entails that $P \circ G_{\alpha} = 0$ and, equivalently, $G_{\alpha} \circ P = 0$. 

In \cite{Dappiaggi:2016fwc}, the commutator function is constructed for all admissible boundary conditions. More precisely, starting from \eqref{eq:conformally_rescaled_dynamics}, we look for  $G_{\bH,\alpha}^\pm \in \mathcal{D}'\big(\mathring{\bH}^{d+1} \times \mathring{\bH}^{d+1}\big)$ fulfilling the same properties of \eqref{advanced_minus_retarded} with $P$ replaced by $P_\eta$, the operator defined in \eqref{eq:conformally_rescaled_dynamics}. The propagators in $\PAdS_{d+1}$ can be reconstructed via the rescaling
\begin{equation}\label{rescaling}
G_{\alpha}^\pm = (zz^\prime)^{\frac{d-1}{2}} G^\pm_{\bH,\alpha} \, .
\end{equation}
The integral kernel of $G^\pm_{\bH,\alpha}$ is of the form \cite{Fulling:1989}
\begin{align}
	G_{\bH,\alpha}^+(\underline{x}-\underline{x}',z,z') &= -\Theta(t-t^\prime) \, G_{\bH,\alpha}(\underline{x}-\underline{x}',z,z') \, ,\label{retarded} \\ 
	G_{\bH,\alpha}^-(\underline{x}-\underline{x}',z,z') &= \Theta(t^\prime-t) \, G_{\bH,\alpha}(\underline{x}-\underline{x}',z,z') \, , \label{advanced}
\end{align}
where $\underline{x}=(t,x_1,\ldots,x_{d-1})$. Here, $G_{\bH,\alpha}(\underline{x}-\underline{x}^\prime,z,z^\prime)$ enjoys the following properties:
\begin{gather}
\left(P_\eta\otimes\mathbb{I}\right)G_{\bH,\alpha} = \left(\mathbb{I}\otimes P_\eta\right)G_{\bH,\alpha} =0 \, , \label{defining_Green} \\
G_{\bH,\alpha}(f,f') = - G_{\bH,\alpha}(f',f) \quad \forall f,f' \in C^\infty_0\big(\mathring{\bH}^{d+1}\big) \, , \notag \\
\iota^*_{t=t'} G_{\bH,\alpha}=0 \, , \quad \textrm{and} \quad 
\iota^*_{t=t'} \partial_{t}G_{\bH,\alpha} = \iota^*_{t=t'} \partial_{t'}G_{\bH,\alpha} = \prod_{i=1}^{d-1} \delta(x_i-x_i') \delta(z-z') \, , \notag
\end{gather}
where $\iota_{t=t'}: \mathring{\bH}^d\times\mathring{\bH}^d \to \mathring{\bH}^{d+1}\times\mathring{\bH}^{d+1}$ is the diagonal map which embeds each factor $\mathring{\bH}^d$ in $\mathring{\bH}^{d+1}$ as the locus $\{t\}\times\mathring{\bH}^d$. In order to construct an explicit expression for the integral kernel of $G_{\bH,\alpha}$, we resort to a mode expansion
\begin{gather*}
G_{\bH,\alpha}(\underline{x}-\underline{x}^\prime,z,z^\prime) = \int_{\bR^d} \frac{\dd^d\underline{k}}{(2\pi)^{\frac{d}{2}}} e^{i\underline{k}\cdot_d(\underline{x}-\underline{x}^\prime)} \, \widehat{G}_{\underline{k},\alpha}(z,z^\prime)
\end{gather*}
where $\underline{k} = (\omega, k_1, \ldots, k_{d-1})$ and $\cdot_d$ stands for the scalar product on the $d$-dimensional Minkowski spacetime. The remaining unknown $\widehat{G}_{\underline{k},\alpha}(z,z^\prime)$ is a symmetric solution of the following set of equations
\begin{equation*}
	(L\otimes\mathbb{I}) \widehat{G}_{\underline{k},\alpha}(z,z^\prime) = (\mathbb{I}\otimes L)\widehat{G}_{\underline{k},\alpha}(z,z^\prime) = q^2 \, \widehat{G}_{\underline{k},\alpha}(z,z^\prime) \, , \quad 
	L \doteq -\frac{\dd^2}{\dd z^2}+\frac{m^2}{z^2} \, ,
\end{equation*}
where $q^2 = \underline{k}\cdot_d\underline{k}$. 

In the following proposition we recollect some results proven in \cite{Dappiaggi:2016fwc}. Letting 
\begin{equation}
	I_{\epsilon}(q,r,t,t') \doteq \int_0^\infty \dd k \, k \left(\frac{k}{r}\right)^{\frac{d-3}{2}} \! J_{\frac{d-3}{2}}(kr) \, q \, \frac{\sin\left(\sqrt{k^2+q^2}(t-t'-i\epsilon)\right)}{\sqrt{2\pi(k^2+q^2)}} \, ,
\end{equation}
where $r^2 \doteq \sum_{i=1}^{d-1} \big(x^i-{x'}^i\big)^2$, the following holds true.

\begin{proposition} \label{prop:GHmodes}
	The causal propagator $G_{\bH,\alpha} \in \mathcal{D}^\prime\big(\mathring{\bH}^{d+1}\times\mathring{\bH}^{d+1}\big)$ for different values of $\nu \in [0,\infty)$ has integral kernel given by the following expressions.
	\begin{enumerate}
		\item If $\nu \in [1, \infty)$,
		\begin{equation*} \label{eq:2pfnularge}
			G_{\bH,\pi}(x,x')
			= \lim_{\epsilon \to 0^+} \sqrt{zz'} \int_0^\infty \dd q \, I_{\epsilon}(q,r,t,t') \, J_\nu(qz) J_\nu(qz') \, .
		\end{equation*}
		\item If $\nu \in (0,1)$ and $c_{\alpha} \doteq \cot(\alpha) \leq 0$, that is, $\alpha\in [\frac{\pi}{2},\pi]$,
		\begin{align*}
		G_{\bH,\alpha}(x,x')
		&= \lim_{\epsilon \to 0^+} \sqrt{zz'}
		\int_0^\infty \dd q \, I_{\epsilon}(q,r,t,t')  \, 
		\frac{\psi_{c_{\alpha}}(z) \psi_{c_{\alpha}}(z')}{c_{\alpha}^2-2c_{\alpha}q^{2\nu}\cos(\nu\pi)+q^{4\nu}} \, ,
		\end{align*}
		where $\psi_{c_{\alpha}}(z) = c_{\alpha}J_\nu(qz)-q^{2\nu}J_{-\nu}(qz)$.
	\end{enumerate}
\end{proposition}

\begin{remark}
	If $\nu=0$, the causal propagator is
	\begin{align*}
	G_{\bH,\alpha}(x,x')
	&= \lim_{\epsilon \to 0^+} \sqrt{zz'}
	\int_0^\infty \dd q \, \left\{ I_{\epsilon}(q,r,t,t') \, 
	\frac{\psi_{c_{\alpha}}(z) \psi_{c_{\alpha}}(z')}{(c_{\alpha}+\frac{2}{\pi}\log(q))^2+1} \right\} \\
	&\quad + 2 I_{\epsilon}\big({-e^{-\pi c_{\alpha}/2}},r,t,t'\big) \, K_0\big(e^{-\pi c_{\alpha}/2}z\big) K_0\big(e^{-\pi c_{\alpha}/2}z'\big) \, ,
	\end{align*}
	where $\psi_{c_{\alpha}}(z) = (c_{\alpha}+\frac{2}{\pi}\log(q))J_0(qz)-Y_0(qz)$, and if $\nu \in (0,1)$ with $c_{\alpha}>0$, it is
	\begin{align*}
	G_{\bH,\alpha}(x,x')
	&= \lim_{\epsilon \to 0^+} \sqrt{zz'}
	\int_0^\infty \dd q \, \left\{ I_{\epsilon}(q,r,t,t') \, 
	\frac{\psi_{c_{\alpha}}(z) \psi_{c_{\alpha}}(z')}{c_{\alpha}^2-2c_{\alpha}q^{2\nu}\cos(\nu\pi)+q^{4\nu}} \right\} \\
	&\quad + 2 I_{\epsilon}\big({-c_{\alpha}^{1/(2\nu)}},r,t,t'\big) \, K_{\nu}\big(c_{\alpha}^{1/(2\nu)}z\big) K_{\nu}\big(c_{\alpha}^{1/(2\nu)}z'\big) \, ,
	\end{align*}
	where $\psi_{c_{\alpha}}(z) = cJ_\nu(qz)-q^{2\nu}J_{-\nu}(qz)$. The second line of each expression gives the contribution of the ``bound state'' mode solutions, as first noted in \cite{Dappiaggi:2016fwc}. As we shall show, there is no ground state in these cases, and hence they will not be considered further in this paper.
\end{remark}

\subsubsection{Structural Properties of the causal propagator}

The explicit expression of the integral kernel of the causal propagator established in Proposition~\ref{prop:GHmodes} does not suffice. Since neither $\PAdS_{d+1}$ nor $\mathring{\bH}^{d+1}$ are globally hyperbolic, we cannot conclude a priori from the standard properties of normally hyperbolic operators that the associated advanced and retarded fundamental solutions possess the support properties required in \eqref{advanced_minus_retarded}, nor that $G_{\bH,\alpha}$ maps smooth, compactly supported initial data into smooth solutions, regardless of the value of $\nu$. Therefore, we investigate the structural properties of $G_{\bH,\alpha}$. We shall focus mainly on the scenario in which $\nu\in(0,1)$ as several boundary conditions are allowed. If $\nu\geq 1$ or if we consider $\alpha=\pi$ for $\nu\in(0,1)$ (corresponding to Dirichlet boundary conditions), our findings agree with those of \cite{Wrochna:2016ruq}. As we have stressed in the introduction, we will not make use of the tools proper of b-calculus, although they are certainly applicable to the case in hand. Special mention should go to the case $\nu=\frac{1}{2}$, which generalize the results of \cite{Dappiaggi:2014gea} to arbitrary boundary conditions. 

Our starting point is still \cite{Dappiaggi:2016fwc}, in which, using Proposition \ref{prop:GHmodes} and \eqref{rescaling},  we show\footnote{More precisely the results of \cite{Dappiaggi:2016fwc} have been proven for the two-point function of the ground state, of which the causal propagator is the antisymmetric part.} that the causal propagator $G_{\alpha} \in \mathcal{D}^\prime(\PAdS_{d+1}\times\PAdS_{d+1})$ for $\nu \in (0,1)$ and $\alpha\in[\frac{\pi}{2},\pi]$ may be written as 
\begin{equation}\label{causal_propagator}
G_{\alpha}(u) = \mathcal{N}_{\alpha} \left[ \cos(\alpha) \, G^{({\rm D})}(u) + \sin(\alpha) \, G^{({\rm N})}(u) \right] \, ,
\end{equation}
where $\mathcal{N}_{\alpha}$ is a normalization constant,
\begin{equation*}
\mathcal{N}_{\alpha} \doteq \frac{\Gamma \left(\frac{d-1}{2}\right)}{2^{d+1}\pi^{\frac{d+1}{2}} \left[\cos(\alpha)+\sin(\alpha)\right]} \, ,
\end{equation*}
and
\begin{subequations} 
	\begin{align}
	G^{({\rm D})}(u) &= \lim_{\epsilon \to 0^+} \left[ u_{\epsilon}^{-\frac{d}{2}-\nu} \, \frac{F \big(\tfrac{d}{2}+\nu, \tfrac{1}{2}+\nu; 1+2\nu; u_{\epsilon}^{-1}\big)}{\Gamma(1+2\nu)} - (\epsilon \leftrightarrow - \epsilon) \right] \, ,\label{GD} \\
	G^{({\rm N})}(u) &=  \lim_{\epsilon \to 0^+} \left[ u_{\epsilon}^{-\frac{d}{2}+\nu} \, \frac{F \big(\tfrac{d}{2}-\nu, \tfrac{1}{2}-\nu; 1-2\nu; u_{\epsilon}^{-1}\big)}{\Gamma(1-2\nu)}  - (\epsilon \leftrightarrow - \epsilon) \right] \, ,\label{GN}
	\end{align}
\end{subequations}
$F$ is the Gaussian hypergeometric function and $u_{\epsilon} \doteq u(\sigma + 2i \epsilon (t-t') + \epsilon^2)$, with
\begin{equation}\label{u-function}
u = u(\sigma) \doteq \cosh^2 \left(\frac{\sqrt{2\sigma}}{2}\right) \, ,
\end{equation}	
where $\sigma$ is the Synge's world function on $\PAdS_{d+1}$ as defined in Section~\ref{sec:AdS}. 

\begin{remark}
	Observe that, if $\nu\geq 1$, \eqref{causal_propagator} is still valid but $\alpha$ must be fixed to $\alpha=\pi$ (or equivalently to $\alpha=0$). The superscripts (D) and (N), respectively in \eqref{GD} and in \eqref{GN}, refer to the fact that, up to a multiplicative constant, the causal propagator coincides with $G^{({\rm D})}$ and with $G^{({\rm N})}$ when we choose either Dirichlet or Neumann boundary conditions.
\end{remark}

The first question that we wish to answer concerns the singular structure of the causal propagator and its dependence on the choice of boundary conditions. To this end, we need a convenient decomposition of $G_\alpha$.

\begin{lemma}\label{lem:convenient_rewriting}
	The causal propagator $G_{\alpha} \in \mathcal{D}^\prime(\PAdS_{d+1}\times\PAdS_{d+1})$, in the case $\nu \in (0,1) \setminus \{\frac{1}{2}\}$, can be decomposed as
	\begin{equation} \label{eq:GG0}
	G_{\alpha}(u) = \mathcal{N}_{\alpha} \left[A_{\alpha} \tilde{G}(u) + B_{\alpha} \iota_z \tilde{G}(u) \right]
	\end{equation}
	where $\iota_z$ is the map \eqref{inversion}, while
	\begin{equation*}
	\tilde{G}(u) = \lim_{\epsilon \to 0^+} \left[F \big(\tfrac{d}{2}+\nu, \tfrac{d}{2}-\nu; \tfrac{d+1}{2}; u_{\epsilon}\big) - (\epsilon \leftrightarrow - \epsilon) \right] \, ,
	\end{equation*}
	and 
	\begin{align*}
	A_{\alpha} &= \cos(\alpha) \Upsilon_A(\nu) + \sin(\alpha) \Upsilon_A(-\nu) \, , \\
	B_{\alpha} &= \cos(\alpha) \Upsilon_B(\nu) + \sin(\alpha) \Upsilon_B(-\nu) \, ,
	\end{align*}
	and
	\begin{align*}
	\Upsilon_A(\nu) \doteq \frac{(-1)^{\frac{d}{2}}}{2^{\frac{d-1}{2}}} 
	\frac{\Gamma \left(\frac{d}{2}+\nu\right) \Gamma \left(\frac{d}{2}-\nu\right) \Gamma \left(1+2\nu\right)}{\Gamma \left(\frac{d+1}{2}\right) \Gamma \left(\frac{d-1}{2}\right)} 
	= - (-1)^{\frac{d}{2}-\nu} \, \Upsilon_B(\nu) \, .
	\end{align*}
\end{lemma}

\begin{proof}
	Using Eq.~(9.132) of \cite{Gradshteyn}, when $\nu \neq \frac{1}{2}$, it is possible to write
	\begin{align*}
	F \big(\tfrac{d}{2}+\nu, \tfrac{d}{2}-\nu; \tfrac{d+1}{2}; & u_{\epsilon}\big)  = A_{\nu} \, (-1)^{-\frac{d}{2}-\nu} \, u_{\epsilon}^{-\frac{d}{2}-\nu} \, \frac{F \big(\tfrac{d}{2}+\nu, \tfrac{1}{2}+\nu; 1+2\nu; u_{\epsilon}^{-1}\big)}{\Gamma(1+2\nu)} \\
	&\quad + A_{-\nu} \, (-1)^{-\frac{d}{2}+\nu} \, u_{\epsilon}^{-\frac{d}{2}+\nu} \, \frac{F \big(\tfrac{d}{2}-\nu, \tfrac{1}{2}-\nu; 1-2\nu; u_{\epsilon}^{-1}\big)}{\Gamma(1-2\nu)} \, , 
\end{align*}
\begin{align*}
	F \big(\tfrac{d}{2}+\nu, \tfrac{d}{2}-\nu; \tfrac{d+1}{2}; 1-u_{\epsilon}\big) &= A_{\nu} \, u_{\epsilon}^{-\frac{d}{2}-\nu} \, \frac{F \big(\tfrac{d}{2}+\nu, \tfrac{1}{2}+\nu; 1+2\nu; u_{\epsilon}^{-1}\big)}{\Gamma(1+2\nu)} \\
	&\quad + A_{-\nu} \, u_{\epsilon}^{-\frac{d}{2}+\nu} \, \frac{F \big(\tfrac{d}{2}-\nu, \tfrac{1}{2}-\nu; 1-2\nu; u_{\epsilon}^{-1}\big)}{\Gamma(1-2\nu)} \, ,
	\end{align*}
	where
	\begin{equation*}
	A_{\nu} \doteq \frac{\Gamma \left(\frac{d+1}{2}\right) \Gamma \left(-2\nu\right)}{\Gamma \left(\frac{d}{2}-\nu\right) \Gamma \left(\frac{1}{2}-\nu\right)} \, .
	\end{equation*}
	Note that if $\nu = \frac{1}{2}$, we can use Eq.~(9.154) of \cite{Gradshteyn} in a similar way and obtain analogous expressions. Inverting these relations, we obtain the desired result.
\end{proof}

As a consequence of this lemma, we can infer the wavefront set of $G_{\alpha}$. We recall that $\PAdS_{d+1}$ together with \eqref{eq:Poincare_metric} is conformal to $\mathring{\bH}^{d+1}\subset\bR^{d+1}$ endowed with the Minkowski metric.

\begin{corollary}[Wavefront set of the causal propagator]\label{WF_G_H}
	Let $G_{\alpha}$ be the causal propagator \eqref{causal_propagator} of the Klein-Gordon operator \eqref{eq:dynamicsP} and let $G_{\bH,\alpha}$ be the one of $P_\eta$ so that, according to \eqref{rescaling}, $G_{\alpha} = (zz^\prime)^{\frac{d-1}{2}} G_{\bH,\alpha}$. Then $WF(G_{\alpha}) = WF(G_{\bH,\alpha})$ and
	\begin{align*}
	WF(G_{\bH,\alpha}) = \left\{ (x,k;x',k') \in T^*\big(\mathring{\bH}^{d+1}\times\mathring{\bH}^{d+1}\big) \setminus \{0\} : 
	(x,k) \sim_{\pm} (x',k') \right\} \, ,
	\end{align*}
	where $\sim_{\pm}$ means that there exist null geodesics with respect to the Minkowski metric $\gamma, \gamma^{(-)} : [0,1] \to\bR^{d+1}$ such that either
	$\gamma(0) = x = (\underline{x}, z)$ and $\gamma(1)=x'$ or 
	$\gamma^{(-)}(0) = x^{(-)} = (\underline{x}, -z)$ and $\gamma^{(-)}(1) = x'$. Moreover,
	$k = (k_{\underline{x}}, k_z)$ ($k^{(-)} = (k_{\underline{x}}, -k_z)$) is coparallel to $\gamma$ ($\gamma^{(-)}$) at 0; and
	$-k'$ is the parallel transport of $k$ ($k^{(-)}$) along $\gamma$ ($\gamma^{(-)}$) at 1 (see Fig.~\ref{fig:wavefrontset}).
\end{corollary}

%
\begin{figure}
	\centering
		{\small
			\begin{tikzpicture}[scale=1]
			
			\draw (-2,-0.12) -- (-2,0.12);
			\draw (2,-0.12) -- (2,0.12);
			\draw[-Latex] (-4,0) -- (6,0)
			node[pos=0.2, below, yshift=-1ex]{$-z_0$}
			node[pos=0.2, above, yshift=1ex]{$x^{(-)}$}
			node[pos=0.6, below, yshift=-1ex]{$z_0$}
			node[pos=0.6, above, yshift=1ex]{$x$}
			node[pos=1, below, xshift=-0.5ex, yshift=-1ex]{$z$}
			;
			\draw[-Latex] (0,-0.6) -- (0,3.6)
			node[pos=1, left, xshift=-0.5ex, yshift=-1ex]{$t$}	
			;
			\draw [fill] (2,0) circle [radius=0.05]; 
			\draw [fill] (-2,0) circle [radius=0.05]; 
			\draw [fill] (1,3) circle [radius=0.05]; 
			\draw [dashed] (4,2) -- (2,0) -- (0,2) -- (-2,0);
			\draw [dashed] (0,2) -- (1,3)
			node[pos=1, above, yshift=0.7ex]{$x'$}	
			;
			
			\draw[-Latex] [thick] (2,0) -- (1.5,0.5)
			node[pos=1, left, xshift=-0.5ex, yshift=-0.5ex]{$k$}
			;
			\draw[-Latex] [thick] (-2,0) -- (-1.5,0.5)
			node[pos=1, right, xshift=0.5ex, yshift=-0.5ex]{$k^{(-)}$}
			;
			\draw[-Latex] [thick] (1,3) -- (0.5,2.5)
			node[pos=1, right, xshift=0.5ex, yshift=-0.2ex]{$-k'$}
			;
			
			\end{tikzpicture}
		}
	\caption{\label{fig:wavefrontset}Wavefront set of the causal propagator.}
\end{figure}
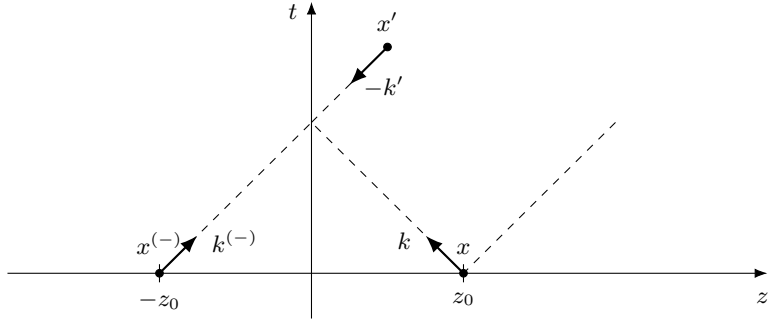

\begin{proof}
	Consider the causal propagator $G_{\alpha}$ written as in \eqref{eq:GG0}. The singular support of $\tilde{G}$ is at $u=1$, which corresponds to $\sigma=0$ or $\sigma_{\bH} = 0$, c.f.~Eq.~\eqref{eq:sigmasigmaM}, whereas the singular support of $\iota_z \tilde{G}$ is at $u=0$, which corresponds to $\sigma^{(-)}=0$ or $\sigma^{(-)}_{\bH}=0$. Every pair of points in the singular support of $\tilde{G}$ is thus connected by a null geodesic and such geodesic, including the endpoints, is contained in a globally hyperbolic subregion of $\PAdS_{d+1}$. An equivalent statement may be made for $\tilde{G}_{\bH} \doteq (zz^\prime)^{\frac{1-d}{2}} \tilde{G}$. Hence, the standard result, c.f.~\cite[Th. 16]{Bar:2009zzb}, applies,
	\begin{align*}
	WF(\tilde{G}_{\bH}) = \left\{ (x,k;x',k') \in T^*\big(\mathring{\bH}^{d+1}\times\mathring{\bH}^{d+1}\big) \setminus \{0\} : 
	(x,k) \sim_+ (x',k') \right\} \, ,
	\end{align*}
	where $\sim_+$ means that there exist null geodesics $\gamma : [0,1] \to \bR^{d+1}$ with 
	$\gamma(0) = x = (\underline{x}, z)$ and $\gamma(1) = x'$; 
	$k = (k_{\underline{x}}, k_z)$ is coparallel to $\gamma$ at 0; and	$-k'$ is the parallel transport of $k$ along $\gamma$ at 1. 
	
	In order to analyse $\iota_z(\tilde{G}_{\bH})$, it suffices to observe that $\iota_z$ leaves both $P$ and $P_\eta$ invariant. Hence, from the standard properties of the wavefront set \cite{Hormander1},
	\begin{align*}
	WF(\iota_z \tilde{G}_{\bH}) = \left\{ (x,k;x',k') \in T^*\big(\mathring{\bH}^{d+1}\times\mathring{\bH}^{d+1}\big) \setminus \{0\} : 
	(x,k) \sim_- (x',k') \right\} \, ,
	\end{align*}
	where $\sim_-$ means that there exist null geodesics $\gamma^{(-)} : [0,1] \to \bR^{d+1}$ with 
	$\gamma^{(-)}(0) = x^{(-)} = (\underline{x}, -z)$ and $\gamma^{(-)}(1) = x'$; 
	$k^{(-)} = (k_{\underline{x}}, -k_z)$ is coparallel to $\gamma^{(-)}$ at 0; and	$-k'$ is the parallel transport of $k^{(-)}$ along $\gamma^{(-)}$ at 1. 
	
	To conclude the proof, observe that $P_\eta$ and $G_{\bH,\alpha}$ are constructed respectively from $P$ and $G_{\alpha}$ via a conformal transformation which, thus, does not change the form of the wavefront set of $G_{\alpha}$.
\end{proof}

This corollary proves also that the wavefront set of the causal propagator is the same for every possible boundary condition and its form agrees both with the one in \cite{Wrochna:2016ruq}, valid for all values of $\nu$ and for the Dirichlet boundary condition, and with the one used in \cite{Dappiaggi:2014gea} for the special case of $\nu=\frac{1}{2}$, still with the Dirichlet boundary condition. 

Since our goal is to study both the space of classical, dynamical configurations and that of observables, our next step consists of fixing the kinematic arena. In this case it is more convenient to work with functions on $\mathring{\bH}^{d+1}$. Recall that their counterpart on $\PAdS_{d+1}$ can be obtained multiplying the pre-factor $z^{\frac{d-1}{2}}$.

\begin{definition}\label{kinematic_configuration}
	We call space of \emph{kinematic/off-shell configurations} with an $\alpha$-boundary condition 
	$$ \mathcal{C}_{\alpha} \big(\mathring{\bH}^{d+1}\big) \doteq \left\{\Phi_{\alpha} \in C^\infty\big(\mathring{\bH}^{d+1}\big)\;\big|\;\widehat{\Phi}_{\underline{k},\alpha} \in D_{\rm max}(L;\alpha) \right\} \, , $$
	where
	$$\widehat{\Phi}_{\underline{k},\alpha} \equiv \widehat{\Phi}_{\underline{k},\alpha}(z) = \int_{\bR^n} \frac{\dd^d\underline{x}}{(2\pi)^\frac{d}{2}} \, e^{-i\underline{k}\cdot_d\underline{x}} \, \Phi_{\alpha}(\underline{x},z) \, , $$
	and $\underline{x}=(t,x_1,...,x_{d-1})$ and $\underline{k}=(\omega,k_1,...,k_{d-1})$. A convenient subspace is
	\begin{align*}
	\tilde{\mathcal{C}}_{\alpha} \big(\mathring{\bH}^{d+1}\big) &\doteq \left\{\Phi_{\alpha} \in \mathcal{C}_\alpha \big(\mathring{\bH}^{d+1}\big)\;\big|\; \exists \, F_1,F_2\in C^\infty\big(\bH^{d+1}\big) \, , \right. \\
	&\qquad\! \left.\Phi_{\alpha}(\underline{x}, z) = \cos(\alpha) z^{\nu+\frac{1}{2}} F_1(\underline{x},z) + \sin(\alpha) z^{-\nu+\frac{1}{2}} F_2(\underline{x},z)\right\} \, ,
	\end{align*}
	where $\alpha\in (0,\pi]$, $\bH^{d+1}\subset\bR^{d+1}$ is the collection of all points in $\bR^{d+1}$ for which $z\geq 0$ and, in the second line, both functions $F_1,F_2$ are restricted to $\mathring{\bH}^{d+1}$.
\end{definition}

Additionally, we introduce distinguished subspaces of $\tilde{\mathcal{C}}_\alpha(\mathring{\bH}^{d+1})$, the first of which adapts to the case in hand the concept of timelike compact functions --- see, for example, \cite{Baernew} or \cite[Ch. 3]{Brunetti:2015vmh}. Let $\pi_{\perp}:\bH^{d+1}\to\bH^d$ be the canonical projection on the hyperplane orthogonal to the $z$-direction. We call
\begin{align*}
\tilde{\mathcal{C}}_{\alpha,{\rm fc}} \big(\mathring{\bH}^{d+1}\big) &\doteq \left\{\Phi_{\alpha} \in \tilde{\mathcal{C}}_\alpha \big(\mathring{\bH}^{d+1}\big)\,\big|\,\right.\\
&\qquad \left. \pi_{\perp}\!\left(\textrm{supp}(\Phi_{\alpha})\cap J^+_{\mathring{\bH}^{d+1}}(p)\right)\textrm{is compact or}\;\emptyset\;\;\forall p\in\mathring{\bH}^{d+1} \right\} \, , \\
\tilde{\mathcal{C}}_{\alpha,{\rm pc}} \big(\mathring{\bH}^{d+1}\big) &\doteq \left\{\Phi_{\alpha} \in \tilde{\mathcal{C}}_\alpha \big(\mathring{\bH}^{d+1}\big)\,\big|\,\right.\\
&\qquad \left. \pi_{\perp}\!\left(\textrm{supp}(\Phi_{\alpha})\cap J^-_{\mathring{\bH}^{d+1}}(p)\right)\textrm{is compact or}\;\emptyset\;\;\forall p\in\mathring{\bH}^{d+1} \right\} \, , 
\end{align*}
while
\begin{align*}
\tilde{\mathcal{C}}_{\alpha,{\rm tc}} \big(\mathring{\bH}^{d+1}\big) &\doteq \tilde{\mathcal{C}}_{\alpha,{\rm fc}} \big(\mathring{\bH}^{d+1}\big) \cap \tilde{\mathcal{C}}_{\alpha,{\rm pc}} \big(\mathring{\bH}^{d+1}\big) \, ,
\end{align*}
where the subscripts pc, fc and tc stand respectively for past compact, future compact and timelike compact. Here $J^\pm_{\mathring{\bH}^{d+1}}$ are the causal future/past with respect to the Minkowski metric. Secondly, we define 
\begin{align}
\tilde{\mathcal{C}}_{\alpha,0} \big(\mathring{\bH}^{d+1}\big) &\doteq \left\{\Phi_{\alpha} \in \mathcal{C}_\alpha \big(\mathring{\bH}^{d+1}\big)\;\big|\; \exists \, F_1,F_2\in C^\infty_0\big(\bH^{d+1}\big) \,\textrm{such that}\;\forall z>0 , \right. \notag \\
&\qquad\! \left. \Phi_{\alpha}(\underline{x}, z) = \cos(\alpha) z^{\nu+\frac{1}{2}} F_1(\underline{x},z) + \sin(\alpha) z^{-\nu+\frac{1}{2}} F_2(\underline{x},z)\right\} \, .\label{eq:observables}
\end{align}
Observe that $C_0^{\infty}\big(\mathring{\bH}^{d+1}\big) \subset \tilde{\mathcal{C}}_{\alpha,0}\big(\mathring{\bH}^{d+1}\big)\subset\tilde{\mathcal{C}}_\alpha\big(\mathring{\bH}^{d+1}\big)$. An example of the support of an element of $\tilde{\mathcal{C}}_{\alpha,0}\big(\mathring{\bH}^{d+1}\big) \setminus C_0^{\infty}\big(\mathring{\bH}^{d+1}\big)$ is in Fig.~\ref{fig:suppf-boundary}.

%
\begin{figure}
	\centering
	{\small
		\begin{tikzpicture}[scale=1]
		
		\path [fill=lightgray] (0,-2) -- (0,2) to (0,-2) arc(-90:90:2); 
		\draw[gray] (0,-2) arc(-90:90:2) --cycle;
		\draw[-Latex] (-0.5,0) -- (5,0)
		node[pos=1, below, xshift=-0.5ex, yshift=-1ex]{$z$}
		node[pos=0.275, above]{${\rm supp}(f_{\alpha})$}
		;			;
		\draw[-Latex] (0,-3) -- (0,3)
		node[pos=1, left, xshift=-0.5ex, yshift=-1ex]{$t$}	
		;			
		\end{tikzpicture}
	}
	\caption{\label{fig:suppf-boundary}Support of some $f_{\alpha} \in \tilde{\mathcal{C}}_{\alpha,0}\big(\mathring{\bH}^{d+1}\big) \setminus C_0^{\infty}\big(\mathring{\bH}^{d+1}\big)$.}
\end{figure}
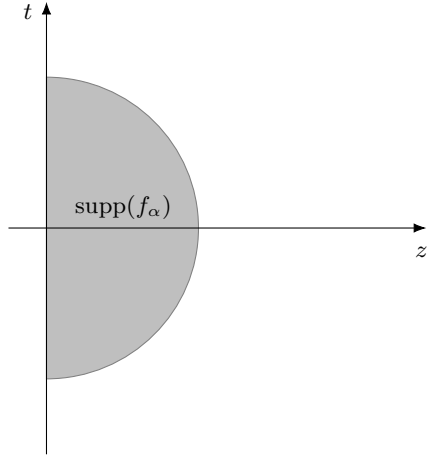

We have all ingredients to prove the structural properties of $G_{\bH,\alpha}$. We start by looking at a distinguished domain, namely smooth and compactly supported functions on $\mathring{\bH}^{d+1}$. In particular, we observe that $C^\infty_0\big(\mathring{\bH}^{d+1}\big) \subset \tilde{\mathcal{C}}_{\alpha,{\rm tc}}\big(\mathring{\bH}^{d+1}\big)$ for all admissible values of $\alpha$. 

\begin{lemma} \label{lemma:KerGH0}
	Let $G_{\bH,\alpha}$ be the bidistribution defined in Proposition \ref{prop:GHmodes}. Then, $G_{\bH,\alpha} : C^\infty_0\big(\mathring{\bH}^{d+1}\big) \to C^{\infty}\big(\mathring{\bH}^{d+1}\big)$,
	and ${\rm supp}(G^\pm_{\bH,\alpha}(f))\subseteq J^\pm_{\mathring{\bH}^{d+1}}({\rm supp}(f))$ for all $f\in C^\infty_0\big(\mathring{\bH}^{d+1}\big)$.
\end{lemma}

\begin{proof}
		The proof is divided in two parts:
		
		\vskip .1cm
		
		\noindent{\em 1.} Since $G_{\bH,\alpha} \in \mcd'\big(\mathring{\bH}^{d+1}\times\mathring{\bH}^{d+1}\big)$, we need to check that partial evaluation on $C^\infty_0\big(\mathring{\bH}^{d+1}\big)$ yields a smooth function. To this end it suffices to observe that the wavefront set of $G_{\bH,\alpha}$ coincides with the one of $G_{\alpha}$  built in Corollary \ref{WF_G_H} since $G_{\alpha}$ and $G_{\bH,\alpha}$ are related by a rescaling smooth in $\mathring{\bH}^{d+1}\times\mathring{\bH}^{d+1}$. Hence, we can apply \cite[Th. 8.2.12]{Hormander1} to conclude that $G_{\alpha}\big[C^\infty_0\big(\mathring{\bH}^{d+1}\big)\big]\subset C^\infty\big(\mathring{\bH}^{d+1}\big)$.
		
		\vskip .2cm  
		
		\noindent {\em 2.} In the second part we show that $G^\pm_{\bH,\alpha}$ obey the sought support properties. To this end it suffices to recollect that in \cite{Dappiaggi:2016fwc} we have shown that the integral kernel of $G_{\alpha}$ vanishes for spacelike separated points. Combining this information with properties \eqref{retarded}, \eqref{advanced} and with the definition of the Heaviside function yields the desired result.
\end{proof}

The next step consists of showing that most of the structural properties of $G^\pm_{\bH,\alpha}$ hold true also when the domain is extended to include timelike compact functions.

\begin{lemma} \label{lemma:KerGH}
	Let $G_{\bH,\alpha}$ be the bidistribution defined in Proposition \ref{prop:GHmodes}. Then, $G_{\bH,\alpha} : \tilde{\mathcal{C}}_{\alpha,{\rm tc}}\big(\mathring{\bH}^{d+1}\big) \to  C^\infty\big(\mathring{\bH}^{d+1}\big)$,
	and ${\rm supp}(G^\pm_{\bH,\alpha}(\gamma))\subseteq J^\pm_{\mathring{\bH}^{d+1}}({\rm supp}(\gamma))$ for all $\gamma\in\tilde{\mathcal{C}}_{\alpha,{\rm tc}}\big(\mathring{\bH}^{d+1}\big)$. 
\end{lemma}

\begin{proof}
	As with the previous lemma we divide the proof in two parts:
	
	\vskip .1cm
	
	\noindent{\em 1.} It suffices to show that $G_{\bH,\alpha}$ is well defined on $\tilde{\mathcal{C}}_{\alpha,0}\big(\mathring{\bH}^{d+1}\big)$. 
	Following the same line of reasoning used in \cite{Baernew} on globally hyperbolic spacetimes, if $p\in\mathring{\bH}^{d+1}$, then, fixing any $\gamma\in\tilde{\mathcal{C}}_{\alpha,{\rm tc}}\big(\mathring{\bH}^{d+1}\big)$, $K_{\gamma,p}\doteq(\textrm{supp}(J^+(p))\cup \textrm{supp}(J^-(p)))\cap\textrm{supp}(\gamma)$, there exists $z_0\in(0,\infty)$ and a compact set $K^\prime\in\bR^d$ such that $K_{\gamma,p}\subset (0,z_0]\times K^\prime$. Hence, choose any cut-off function $\eta\in C^\infty\big(\mathring{\bH}^{d+1}\big)$ of the form $\eta=\eta_1(z)\eta_2(\underline{x})$ where $\eta_2\in C^\infty_0\big(\bR^d\big)$ is equal to $1$ in $K^\prime$, while $\eta_1\in C^\infty(0,\infty)$ is $1$ between $(0,z_0]$ and vanishes with all its derivatives for all $z$ greater than a certain value $z_1>z_0$. It holds that 
	 $G_{\bH,\alpha}(\gamma)(p)=G_{\bH,\alpha}(\eta\gamma)(p)$, though $\eta\gamma\in\widetilde{\mathcal{C}}_{\alpha,0}\big(\mathring{\bH}^{d+1}\big)$. Hence, it suffices to prove that  $G_{\bH,\alpha}\big[\widetilde{\mathcal{C}}_{\alpha,0}\big(\mathring{\bH}^{d+1}\big)\big] \subset C^\infty\big(\mathring{\bH}^{d+1}\big)$. 
	 
	 This can be achieved using the explicit expression for $G_{\bH,\alpha}$ constructed in Proposition \ref{prop:GHmodes} and testing it against $f(\underline{x}^\prime,z^\prime)\in \widetilde{\mathcal{C}}_{\alpha,0}\big(\mathring{\bH}^{d+1}\big)$. Since the explicit formula is not particularly enlightening, we highlight the two main points of the calculation. First, the integral along the $\underline{x}^\prime$-direction is controlled by the compactness of $f$ along these directions. And second, the integral in $z^\prime$ is convergent given the assumptions in \eqref{eq:observables} and the asymptotic behavior of the Bessel functions as $z^\prime$ tends to $0$.
	
	\vskip .2cm  
	
	\noindent {\em 2.} To show that $G^\pm_{\bH,\alpha}$ obey the sought support properties, one can use the same reasoning as in Lemma \ref{lemma:KerGH0}. Hence, we shall not repeat it. 
\end{proof}

In the last step we finally characterize the kernel of the causal propagator when acting on smooth and compactly supported functions.

\begin{proposition}\label{prop:KerG}
	Let $G_{\bH,\alpha}$ be the bidistribution defined in Proposition \ref{prop:GHmodes}. Then, $\ker(G_{\bH,\alpha})\big|_{\widetilde{C}_{\alpha,0}(\mathring{\bH}^{d+1})}=P_\eta\big[\widetilde{C}_{\alpha,0}\big(\mathring{\bH}^{d+1}\big)\big]$, while
	$$\ker_0(G_{\bH,\alpha}) \equiv \ker(G_{\bH,\alpha})\big|_{C^\infty_0(\mathring{\bH}^{d+1})} \doteq \left\{f\in C^\infty_0\big(\mathring{\bH}^{d+1}\big) \;|\;  f = P_\eta\phi, \; \phi \in \widetilde{\mathcal{C}}_{\alpha,0}\big(\mathring{\bH}^{d+1}\big) \right\}.$$
\end{proposition}

\begin{proof}
	Suppose that $\phi^\prime \in \widetilde{\mathcal{C}}_{\alpha,0}\big(\mathring{\bH}^{d+1}\big)$ is of the form $P_\eta\phi$, $\phi\in\widetilde{\mathcal{C}}_{\alpha,0}(\mathring{\bH}^{d+1})$. Then, Lemma~\ref{lemma:KerGH} guarantees that $G_{\bH,\alpha}$ can be applied to $P_\eta\phi$, while $G_{\bH,\alpha} \circ P_\eta=0$ entails that $f\in\ker(G_{\bH,\alpha})$. 
	
	Conversely, suppose that $\phi^\prime \in \widetilde{\mathcal{C}}_{\alpha,0}\big(\mathring{\bH}^{d+1}\big)$ is such that $G_{\bH,\alpha}(\phi^\prime)=G^+_{\bH,\alpha}(\phi^\prime)-G^-_{\bH,\alpha}(\phi^\prime)=0$. Then $G^+_{\bH,\alpha}(\phi^\prime)=G^-_{\bH,\alpha}(\phi^\prime)$ and, in view of the properties of the advanced and of the retarded fundamental solution, $\phi^\prime=P_\eta(G^+_{\bH,\alpha}(\phi^\prime))$. The same argument of Lemma \ref{lemma:KerGH} allows us to conclude that $G^+_{\bH,\alpha}(\phi^\prime)$ is smooth in $\mathring{\bH}^{d+1}$. Combining this information with the support properties of $G^\pm_{\bH,\alpha}$ yields that $G^+_{\bH,\alpha}(\phi^\prime)\in\widetilde{\mathcal{C}}_{\alpha,0}\big(\mathring{\bH}^{d+1}\big)$. The explicit expression for $\ker_0(G_{\bH,\alpha})$ descends per restriction from the one of $\ker(G_{\bH,\alpha})\big|_{\widetilde{C}_{\alpha,0}(\mathring{\bH}^{d+1})}$, taking into account that, if one starts from a smooth and compactly supported function $f$, then $\textrm{supp}(G^+_{\bH,\alpha}(f))\cap \textrm{supp}(G^-_{\bH,\alpha}(f))$ is not necessarily compact in view of the geometric properties of $\mathring{\bH}^{d+1}$.
\end{proof}


\subsection{Off-shell algebra}
\label{sec:offshellalgebra}

In this section, we construct the algebra of observables for a real, massive scalar field in $\PAdS_{d+1}$, studying at the same time its structural properties. We use the functional approach, see \cite[\S 3]{Brunetti:2015vmh}, which is well-suited to be adapted to the case of spacetimes with boundaries as proven in the recent investigation of the Casimir effect \cite{Dappiaggi:2014gea}. Furthermore, in view of subsection~\ref{sec:KG_equation}, we focus our attention directly to the upper-half plane $\mathring{\bH}^{d+1}$ and to equation~\eqref{eq:conformally_rescaled_dynamics}. First, we consider regular functionals on the space of off-shell configurations.

\begin{definition}
	Let $F_{\alpha} : \mathcal{C}_\alpha \big(\mathring{\bH}^{d+1}\big) \to \bC$ be any functional and let $U_{\alpha} \subset \mathcal{C}_\alpha \big(\mathring{\bH}^{d+1}\big)$ be an open set. We say that $F_{\alpha}$ is \emph{differentiable of order $k$} if, for all $m=1, \ldots, k$, the $m$-th order G\^ateaux derivatives
	\begin{equation*}
		F_{\alpha}^{(m)}(\phi)(\phi_1, \ldots, \phi_m) 
		\doteq \left.\frac{\partial^m}{\partial \lambda_1 \ldots \partial \lambda_m}\right|_{\lambda_1 = \ldots = \lambda_m = 0} F_{\alpha} \left(\phi + \sum_{j=1}^m \lambda_j \phi_j \right)
	\end{equation*}
	exist as jointly continuous maps from $U_{\alpha} \times \big(\mathcal{C}_\alpha \big(\mathring{\bH}^{d+1}\big)\big)^{\otimes m}$ to $\bC$. Moreover, $F_{\alpha}$ is \emph{smooth} if it is differentiable at all orders $k \in \bN$, and it is \emph{regular} if it is smooth, if for all $k \geq 1$ and for all $\phi \in \mathcal{C}_\alpha \big(\mathring{\bH}^{d+1}\big)$, $F_{\alpha}^{(k)}(\phi) \in \tilde{\mathcal{C}}_{\alpha,0} \big(\mathring{\bH}^{(d+1)k}\big)$ and if only finitely many functional derivatives do not vanish. We denote the set of regular functionals on $\mathring{\bH}^{d+1}$ by $\mathcal{F}_{\alpha} \big(\mathring{\bH}^{d+1}\big)$.
\end{definition}

Secondly, we introduce a suitable product on the set of regular functionals, given by $\star_\alpha : \mathcal{F}_{\alpha} \big(\mathring{\bH}^{d+1}\big) \times \mathcal{F}_{\alpha} \big(\mathring{\bH}^{d+1}\big) \to \mathcal{F}_{\alpha} \big(\mathring{\bH}^{d+1}\big)$,
\begin{equation} \label{eq:algebraprodG}
	(F_{\alpha} \star_\alpha F'_{\alpha})(\phi) = \big(\mathcal{M} \circ \exp(i \Gamma_{G_{\bH,\alpha}})(F_{\alpha} \otimes F'_{\alpha})\big)(\phi) \, ,
\end{equation}
where $F_{\alpha}, \, F'_{\alpha} \in \mathcal{F}_{\alpha} \big(\mathring{\bH}^{d+1}\big)$. Here, $\mathcal{M}$ stands for the pointwise multiplication, that is, $\mathcal{M} (F_{\alpha} \otimes F'_{\alpha})(\phi) \doteq F_{\alpha}(\phi) F'_{\alpha}(\phi)$, and
\begin{equation}
	\Gamma_{G_{\alpha}} \doteq \frac{1}{2} \int_{\mathring{\bH}^{d+1} \times \mathring{\bH}^{d+1}} G_{\bH,\alpha}(x,x') \, \frac{\delta}{\delta \phi(x)} \otimes \frac{\delta}{\delta \phi(x')} \, .
\end{equation}

\begin{definition}
	We call $\mathcal{A}_{\alpha}\big(\mathring{\bH}^{d+1}\big) \equiv \big(\mathcal{F}_{\alpha} \big(\mathring{\bH}^{d+1}\big), \star_\alpha \big)$ the \emph{off-shell $\ast$-algebra} of the system with complex conjugation as $\ast$-operation. It is generated by the functionals
	\begin{equation}\label{eq:generating_functional}
		F_{f_{\alpha}}(\phi) = \int_{\mathring{\bH}^{d+1}} \dd^{d+1}x \, \phi_{\alpha}(x) f_{\alpha}(x) \, ,
	\end{equation}
	where $f_{\alpha} \in \tilde{\mathcal{C}}_{\alpha,0} \big(\mathring{\bH}^{d+1}\big)$ and $\phi_{\alpha} \in \mathcal{C}_\alpha \big(\mathring{\bH}^{d+1}\big)$.
\end{definition}

Note that \eqref{eq:generating_functional} is well defined since Definitions~\ref{kinematic_configuration} and \eqref{eq:observables} guarantee that the behavior of both $\phi_{\alpha}$ and $f_{\alpha}$ as $z \to 0$ is such that the integral along this direction is not divergent. Another question concerns whether $(F_{f_{\alpha}}\star_\alpha F_{f^\prime_{\alpha}})(\phi_{\alpha})$ is finite for arbitrary choices of $\phi_{\alpha} \in \mathcal{C}_\alpha\big(\mathring{\bH}^{d+1}\big)$ and $f_{\alpha}, \, f_{\alpha}^\prime\in\tilde{\mathcal{C}}_{\alpha,0} \big(\mathring{\bH}^{d+1}\big)$. Also in this case one can show by direct computation that this is a by-product of the explicit form of the integral kernel of the causal propagator in Proposition~\ref{prop:GHmodes} together with the behavior at the boundary of the elements of $\tilde{\mathcal{C}}_{\alpha,0} \big(\mathring{\bH}^{d+1}\big)$.


\subsection{On-shell algebra}
\label{sec:onshellalgebra}

In order to construct the on-shell algebra we restrict the allowed configurations in $\mathcal{A}_{\alpha}\big(\mathring{\bH}^{d+1}\big)$ from $\mathcal{C}_\alpha \big(\mathring{\bH}^{d+1}\big)$ to the space of dynamical configurations, $\mathcal{S}_\alpha \big(\mathring{\bH}^{d+1}\big) \doteq \{\phi_{\alpha} \in \mathcal{C}_\alpha \big(\mathring{\bH}^{d+1}\big) \, | \, P_{\eta} \phi_{\alpha} = 0 \}$.

Let $\mathcal{F}^{\rm on}_{\alpha}\big(\mathring{\bH}^{d+1}\big)$ denote the set of all functionals $F_{[f_{\alpha}]}: \mathcal{S}_\alpha \big(\mathring{\bH}^{d+1}\big) \to \bC$, with $[f_{\alpha}] \in \frac{\tilde{\mathcal{C}}_{\alpha,0} (\mathring{\bH}^{d+1})}{P_\eta[\widetilde{C}_{\alpha,0}(\mathring{\bH}^{d+1})]}$, $\ker(G_{\bH,\alpha})$ being characterized in Proposition \ref{prop:KerG}, such that
\begin{equation*}
F_{[f_{\alpha}]}(\phi_{\alpha}) = \int_{\mathring{\bH}^{d+1}} \dd^{d+1}x \,  f_{\alpha}(x) \phi_{\alpha}(x) \, .
\end{equation*}
One can prove by integrating by parts that the right-hand side does not depend on the choice of representative of $[f_{\alpha}]$ on account of the $\alpha$-boundary condition. With a slight abuse of notation, we will identify every generator $F_{[f_{\alpha}]}$ directly with its label $[f_{\alpha}]$, hence writing $[f_{\alpha}] \in \mathcal{F}^{\rm on}_{\alpha}(\mathring{\bH}^{d+1})$.

\begin{proposition}
	The set $\mathcal{F}^{\rm on}_{\alpha}\big(\mathring{\bH}^{d+1}\big)$ is
	\begin{enumerate}
		\item \emph{separating}, that is, for every pair $\phi_{\alpha}, \, \phi'_{\alpha} \in \mathcal{S}_\alpha \big(\mathring{\bH}^{d+1}\big)$ with $\phi_{\alpha}\neq\phi'_{\alpha}$, there exists $[f_{\alpha}] \in \frac{\tilde{\mathcal{C}}_{\alpha,0} (\mathring{\bH}^{d+1})}{P_\eta[\widetilde{C}_{\alpha,0}(\mathring{\bH}^{d+1})]}$ such that $F_{[f_{\alpha}]}(\phi_{\alpha}) \neq F_{[f_{\alpha}]}(\phi'_{\alpha})$.
		\item \emph{symplectic} if endowed with the following weakly non-degenerate symplectic form $\sigma_{\bH} : \mathcal{F}_{\alpha}^{\rm on}\big(\mathring{\bH}^{d+1}\big) \times \mathcal{F}_{\alpha}^{\rm on}\big(\mathring{\bH}^{d+1}\big) \to \bR$,
		\begin{equation}\label{eq:symplectic_form}
		\sigma_{\bH} \big(F_{[f_{\alpha}]},F_{[f'_{\alpha}]}\big) = \int_{\mathring{\bH}^{d+1}} \dd^{d+1}x \, f_{\alpha}(x) G_{\bH,\alpha}(f'_{\alpha})(x) \, .
		\end{equation}
		\item  \emph{optimal}, that is, for every pair $[f_{\alpha}], \, [f'_{\alpha}] \in \frac{\tilde{\mathcal{C}}_{\alpha,0} (\mathring{\bH}^{d+1})}{P_\eta[\widetilde{C}_{\alpha,0}(\mathring{\bH}^{d+1})]}$ with $[f_{\alpha}] \neq [f'_{\alpha}]$, there exists $\phi_{\alpha} \in \mathcal{S}_\alpha \big(\mathring{\bH}^{d+1}\big)$ such that $F_{[f_{\alpha}]}(\phi_{\alpha}) \neq F_{[f'_{\alpha}]}(\phi_{\alpha})$.

	\end{enumerate}
\end{proposition}

\begin{proof}
	Starting with (1), consider $\phi_{\alpha}, \, \phi'_{\alpha} \in \mathcal{S}_\alpha \big(\mathring{\bH}^{d+1}\big)$ such that $\phi_{\alpha} \neq \phi'_{\alpha}$. Then, since $\phi_{\alpha}-\phi'_{\alpha} \in C^{\infty}\big(\mathring{\bH}^{d+1}\big)$ is not vanishing, Hahn-Banach theorem guarantees that $C^{\infty}_{0}(\mathring{\bH}^{d+1})$ is separating for $C^{\infty}\big(\mathring{\bH}^{d+1}\big)$ with respect to the pairing defined via integration. Hence, there exists $f \in C^{\infty}_{0}\big(\mathring{\bH}^{d+1}\big)$ such that $\int_{\mathring{\bH}^{d+1}}d^{d+1}x\, f \, (\phi_{\alpha}-\phi'_{\alpha})$ does not vanish. As $C^{\infty}_{0}\big(\mathring{\bH}^{d+1}\big) \subset \mathcal{C}_{\alpha,0}\big(\mathring{\bH}^{d+1}\big)$, the result follows.
	
	
	Concerning (2), note that, per construction, $\sigma_{\bH}$ is bilinear. Suppose that there exists $F_{[f'_{\alpha}]} \in \mathcal{F}_{\alpha}^{\rm on}\big(\mathring{\bH}^{d+1}\big)$ such that $\sigma_{\bH} \big(F_{[f_{\alpha}]},F_{[f'_{\alpha}]}\big) = 0$ for all $F_{[f_{\alpha}]} \in \mathcal{F}_{\alpha}^{\rm on}\big(\mathring{\bH}^{d+1}\big)$. Then, in view of the separating properties of the integrals over $\mathring{\bH}^{d+1}$, it descends $G_{\bH,\alpha}(f'_{\alpha}) = 0$. On account of Lemma \ref{lemma:KerGH}, since $[f'_{\alpha}] \in \frac{\tilde{\mathcal{C}}_{\alpha,0} (\mathring{\bH}^{d+1})}{P_\eta[\widetilde{C}_{\alpha,0}(\mathring{\bH}^{d+1})]}$, one has $[f'_{\alpha}] = 0$. In other words $\sigma_\bH$ is non-degenerate.
	
	As for the antisymmetry property, first note that, using $P_\eta\circ G_{\bH,\alpha}^\pm=G_{\bH,\alpha}^\pm\circ P_\eta=\mathbb{I}$, $\mathbb{I}$ being the identity on $\widetilde{\mathcal{C}}_{\alpha,0}\big(\mathring{\bH}^{d+1}\big)$, it holds
	\begin{align*}
		\sigma_{\bH} \big(F_{[f_{\alpha}]},F_{[f'_{\alpha}]}\big) &= \int_{\bR^{d}} \dd^{d}x \int_0^{\infty} \dd z \, f_{\alpha}(x) \left[ G_{\bH,\alpha}^+(f'_{\alpha})(x) - G_{\bH,\alpha}^-(f'_{\alpha})(x) \right] \\
		&= \int_{\bR^{d}} \dd^{d}x \int_0^{\infty} \dd z \left[ P_{\eta} G_{\bH,\alpha}^- (f_{\alpha})(x) G_{\bH,\alpha}^+(f'_{\alpha})(x) - (+\leftrightarrow-) \right] \, .
	\end{align*}
	Write $P_{\eta} = \Box_{(d)} + \frac{\dd^2}{\dd z^2} - \frac{m^2}{z^2}$ and define $\Phi^{\pm}_{f_{\alpha}} \doteq G_{\bH,\alpha}^{\pm} (f_{\alpha})$. Then,
	\begin{multline*}
		\int_{\bR^{d}} \dd^{d}x \int_0^{\infty} \dd z \, \left(\Box_{(d)} + \frac{\dd^2}{\dd z^2} - \frac{m^2}{z^2} \right) \Phi^-_{f_{\alpha}}(x) \Phi^+_{f'_{\alpha}}(x) \\
		= \int_{\bR^{d}} \dd^{d}x \left[ \int_0^{\infty} \dd z \, \Phi^-_{f_{\alpha}}(x) \left(\Box_{(d)} + \frac{\dd^2}{\dd z^2} - \frac{m^2}{z^2} \right)  \Phi^+_{f'_{\alpha}}(x)
		+ B_{\alpha} \right] \, ,
	\end{multline*}
	where
	\begin{align*}
		B_{\alpha} \doteq \left. \left( \Phi^+_{f'_{\alpha}}(x) \frac{\dd \Phi^-_{f_{\alpha}}(x)}{\dd z} - \Phi^-_{f_{\alpha}}(x) \frac{\dd \Phi^+_{f'_{\alpha}}(x)}{\dd z} \right) \right|_{z=0^+}^{z=\infty}
		= \left. W_z \big[\Phi^+_{f'_{\alpha}}, \Phi^-_{f_{\alpha}}\big]\right|_{z=0^+}^{z=\infty} \, .
	\end{align*}
	Recall that, since $f_{\alpha} \in \tilde{\mathcal{C}}_{\alpha,0} \big(\mathring{\bH}^{d+1}\big)$, $\lim_{z \to \infty} W_z \big[\Phi^+_{f'_{\alpha}}, \Phi^-_{f_{\alpha}}\big] = 0$ and
	\begin{align*}
		\lim_{z \to 0^+} \left[ \cos(\alpha) W_z \big[\Phi^-_{f_{\alpha}}, \Phi_1 \big] + \sin(\alpha) W_z \big[\Phi^-_{f_{\alpha}}, \Phi_2 \big] \right] &= 0 \, , \\
		\lim_{z \to 0^+} \left[ \cos(\alpha) W_z \big[\Phi^+_{f'_{\alpha}}, \Phi_1 \big] + \sin(\alpha) W_z \big[\Phi^+_{f'_{\alpha}}, \Phi_2 \big] \right] &= 0 \, ,
	\end{align*}
	where we used Definition~\ref{def:boundary_condition}.	After a straightforward calculation we obtain that $\lim_{z \to 0^+} W_z \big[\Phi^+_{f'_{\alpha}}, \Phi^-_{f_{\alpha}}\big] = 0$ and thus $B_{\alpha}=0$. Therefore,
	\begin{align*}
	    \sigma_{\bH} \big(F_{[f_{\alpha}]},F_{[f'_{\alpha}]}\big) &= \int_{\bR^{d}} \dd^{d}x \int_0^{\infty} \dd z \left[ G_{\bH,\alpha}^- f_{\alpha}(x) P_{\eta} G_{\bH,\alpha}^+(f'_{\alpha})(x) - (+\leftrightarrow-) \right] \\
	    &= - \sigma_{\bH} \big(F_{[f'_{\alpha}]},F_{[f_{\alpha}]}\big) \, .	
	\end{align*}
	%
	
	
	As for (3), since $\mathcal{F}_{\alpha}^{\rm on}\big(\bH^{d+1}\big)$ is generated by linear functions, it suffices to show a contradiction in assuming that there exists a non trivial $[f_{\alpha}] \in \frac{\tilde{\mathcal{C}}_{\alpha,0} (\mathring{\bH}^{d+1})}{P_\eta[\widetilde{C}_{\alpha,0}(\mathring{\bH}^{d+1})]}$ such that $F_{[f_{\alpha}]}(\phi) = 0$ for all $\phi_{\alpha} \in \mathcal{S}_\alpha(\mathring{\bH}^{d+1})$. If this were the case, using the antisymmetry property proven in the preceding point, one has $0 = \int_{\mathring{\bH}^{d+1}} f_{\alpha} \, G_{\bH,\alpha}(h) = - \int_{\mathring{\bH}^{d+1}} \, G_{\bH,\alpha}(f_{\alpha}) \, h$ for all $h\in C^\infty_0(\mathring{\bH}^{d+1})$. Since $G_{\bH,\alpha}(f_{\alpha})\in\mcd^\prime(\mathring{\bH}^{d+1})$ and $h$ is an arbitrary test function, $G_{\bH,\alpha}(f_{\alpha})=0$, which is a contradiction.
\end{proof}

It is worth observing that, as a by-product of the proof of the preceding Proposition, we have also shown that $P_\eta$ is a formally self-adjoint operator, a property which often plays a key, technical role. In addition, the following definition is now justified.

\begin{definition}\label{on-shell_algebra}
	We call $\mathcal{A}^{\rm on}_{\alpha}\big(\mathring{\bH}^{d+1}\big) \equiv \big(\mathcal{F}^{\rm on}_{\alpha}\big(\mathring{\bH}^{d+1}\big), \star_\alpha \big)$	the \emph{on-shell $\ast$-algebra} of the system, generated by the functionals $\mathcal{F}^{\rm on}_{\alpha}\big(\mathring{\bH}^{d+1}\big)$, where the $\star_\alpha$-product is the same as in \eqref{eq:algebraprodG}.
\end{definition}

\noindent We now prove some structural properties of the on-shell $\ast$-algebra.

\begin{proposition}\label{Prop:time_slice}
	The algebra $\mathcal{A}^{\rm on}_{\alpha}\big(\mathring{\bH}^{d+1}\big)$ is 
	\begin{enumerate}
		\item causal, that is, algebra elements, supported in spacelike separated regions commute. 
		\item fulfils the {\bf time-slice axiom}, i.e.~let $O_{\epsilon,\bar{t}} \doteq (\bar{t}-\epsilon,\bar{t}+\epsilon) \times \mathring{\bH}^{d}$, $\epsilon>0$ and $\bar{t}\in\bR$ arbitrary and let $\mathcal{A}^{\rm on}_{\alpha}(O_{\epsilon, \bar{t}})$ be the on-shell algebra restricted to $O_{\epsilon, \bar{t}}$, then $\mathcal{A}^{\rm on}_{\alpha}\big(\mathring{\bH}^{d+1}\big) \simeq \mathcal{A}^{\rm on}_{\alpha}(O_{\epsilon,\bar{t}})$.
	\end{enumerate}
\end{proposition}

\begin{proof}
{\em (1)}: Since $\mathcal{A}^{\rm on}_{\alpha}\big(\mathring{\bH}^{d+1}\big)$ is generated by the functionals $\mathcal{F}^{\rm on}_{\alpha}\big(\mathring{\bH}^{d+1}\big)$, we can focus on this set. Causality is a direct consequence of \eqref{eq:algebraprodG}, together with the support properties of $G_{\bH,\alpha}$ as per Lemma \ref{lemma:KerGH}. Hence, algebra elements supported in spacelike separated regions do commute. 
	
\vskip .2cm

\noindent{\em (2)}: For the time-slice axiom to hold true, it suffices to prove it at the level of the generators $\mathcal{F}^{\rm on}_{\alpha}(O_{\epsilon, \bar{t}})$ and $\mathcal{F}^{\rm on}_{\alpha}\big(\mathring{\bH}^{d+1}\big)$. Let $[f_{\alpha}]\in\mathcal{F}^{\rm on}_{\alpha}(O_{\epsilon, \bar{t}})$ and let $\chi(t)$ be any smooth function which is $0$ for $t\geq\bar{t}+\epsilon$ and $1$ for $t\leq\bar{t}-\epsilon$. Consider $g_{\alpha} \doteq P_\eta \chi G_{\bH,\alpha}(f_{\alpha})$. In view of Lemma \ref{lemma:KerGH}, $\textrm{supp}(g_{\alpha})\subset O_{\epsilon, \bar{t}}$ and thus $h$ identifies an equivalence class $[g_{\alpha}]$ and thus an element of $\mathcal{F}_{\alpha}^{\rm on}(O_{\epsilon, \bar{t}})$. On account of the characterization of the kernel of $G_{\bH,_{\alpha}}$ in Lemma \ref{lemma:KerGH}, we have identified a map $\iota_{O,\bH}:\mathcal{F}^{\rm on}_{\alpha}\big(\mathring{\bH}^{d+1}\big)\to\mathcal{F}^{\rm on}_{\alpha}(O_{\epsilon, \bar{t}})$ such that $[f_{\alpha}] \mapsto \iota_{O,\bH}([f_{\alpha}])=[P_\eta\chi G_{\bH,\alpha}(f_{\alpha})]$, which is per construction injective. It is also surjective, since the inverse $\iota_{O,\bH}^{-1}:\mathcal{F}^{\rm on}_{\alpha}(O_{\epsilon, \bar{t}})\to\mathcal{F}^{\rm on}_{\alpha}\big(\mathring{\bH}^{d+1}\big)$ is nothing but $\iota^{-1}_{O, \bH}([g_{\alpha}])=[g_{\alpha}]$, where the right hand side (RHS) has to be interpreted as an equivalence class in $\mathcal{F}^{\rm on}_{\alpha}\big(\mathring{\bH}^{d+1}\big)$ generated by any representative of $[g_{\alpha}]\in\mathcal{F}^{\rm on}_{\alpha}(O_{\epsilon, \bar{t}})$. Hence $\iota_{O,\bH}$ is an isomorphism of vector spaces. In order to promote it to a $*$-isomorphism at the level of algebras, it suffices to notice that the $*$-operation is left untouched by this map and that it preserves the symplectic form \eqref{eq:symplectic_form}. As a matter of fact, calling $G_O$ the restriction of $G_{\bH,\alpha}$ to $O_{\epsilon, \bar{t}}$, then, for any $[f_{\alpha}],[f^\prime_{\alpha}]\in \mathcal{F}^{\rm on}_{\alpha}\big(\mathring{\bH}^{d+1}\big)$, choose two representatives $f_{\alpha}, \, f^\prime_{\alpha}$ whose support lies in $O_{\epsilon, \bar{t}}$. Then, from \eqref{eq:symplectic_form}, it descends
$$\sigma_{\bH}([f_{\alpha}],[f^\prime_{\alpha}]) = \int_{\mathring{\bH}^{d+1}} \dd^{d+1}x \, f_{\alpha} \, G_{\bH,\alpha}(f^\prime_{\alpha}) = \int_{O_{\epsilon,\bar{t}}} \dd^{d+1}x f_{\alpha} \, G_O(f^\prime_{\alpha}) \, , $$
where the RHS is the symplectic form generating the $\star_\alpha$-product of $\mathcal{A}_{\alpha}^{\rm on}(O_{\epsilon,\bar{t}})$.
\end{proof}

Motivated by the analysis in \cite{Kay:1992es}, we investigate an additional property of $\mathcal{A}^{\rm on}_\alpha\big(\mathring{\bH}^{d+1}\big)$. In this paper it is remarked that a natural requirement for a quantum field theory in a non-globally hyperbolic spacetime is the existence of a $*$-isomorphism between the global algebra of observables, restricted to any globally hyperbolic subregion and the counterpart built directly in such a region from the equations of motion, following the standard procedure. This property is also known as {\em F-locality}. To this end, let us consider any globally hyperbolic subregion $D \subset \mathring{\bH}^{d+1}$ and let $G_D$ be the unique causal propagator associated to $P_\eta$. The existence of $G_D$ is guaranteed by the fact that $P_\eta$ is normally hyperbolic in $D$. Let $\mathcal{S}(D)$ be the set of all smooth solutions in $D$ of \eqref{eq:conformally_rescaled_dynamics} and, for any $[\tilde{f}]\in\frac{C^\infty_0(D)}{P_\eta[C^\infty_0(D)]}$, let 
$$F_{[\tilde{f}]}(\phi^\prime) = \int_D \dd^{d+1}x \, \tilde{f}(x) \, \phi^\prime(x) \, , \quad \phi^\prime\in\mathcal{S}(D) \, .$$
We call $\mathcal{A}(D)$ the $*$-algebra generated by these functionals together with a $\star$-product of the same form of \eqref{eq:algebraprodG} though with $G_{\bH,\alpha}$ replaced by $G_D$.

\begin{proposition}
	The algebra $\mathcal{A}_\alpha^{\rm on}(\mathring{\bH}^{d+1})$ is {\bf F-local}, namely it is $*$-isomorphic to $\mathcal{A}(D)$.
\end{proposition}

\begin{proof}
	In view of Corollary \eqref{WF_G_H}, we know that in every globally hyperbolic subregion $D \subset \mathring{\bH}^{d+1}$ $G_{\bH,\alpha}$ and $G_D$ have the same wavefront set. Hence, in view of \cite{Radzikowski:1996pa,Radzikowski:1996ei}, they differ only by a smooth term, i.e.~on $D$ there exists $W_{\alpha}\in C^\infty(D\times D)$ such that $G_{\bH,\alpha}=G_D+W_{\alpha}$. Since, as vector spaces $\mathcal{A}_\alpha^{\rm on}\big(\mathring{\bH}^{d+1}\big)$ and $\mathcal{A}(D)$ coincide, one can adapt to this context the standard analysis of \cite[\S 5.1]{Brunetti:2009pn} to conclude that the sought isomorphism exists.
\end{proof}


\subsection{Hadamard states}
\label{sec:hadamardstates}

We now discuss algebraic states for the system in hand, that is, linear functionals $\omega_{\alpha}^{\bH} : \mathcal{A}^{\rm on}_{\alpha}\big(\mathring{\bH}^{d+1}\big) \to \bC$ for which
\begin{equation*}
	\omega_{\alpha}^{\bH}(1) = 1 \, , \qquad \omega_{\alpha}^{\bH}(a^*a) \geq 0 \, , \, \forall a \in \mathcal{A}^{\rm on}_{\alpha}\big(\mathring{\bH}^{d+1}\big) \, .
\end{equation*}
Notable are the states which are defined in terms of their $n$-point correlation functions and especially useful are the Gaussian/quasifree ones. These are those states for which the odd $n$-point functions vanish and the even ones are of the form
\begin{equation*}
\omega^{\bH}_{\alpha,2n}(f_1 \otimes \ldots \otimes f_n) = \sum_{\pi_{2n} \in S'_{2n}} \prod_{i=1}^n \omega^{\bH}_{\alpha,2} \left(f_{\pi_{2n}(i-1)} \otimes f_{\pi_{2n}(i)}\right) \, ,
\end{equation*}
where $S'_{2n}$ is the set of ordered permutations of $2n$ elements. In addition, we look for those $\omega^{\bH}_{\alpha,2}$ which identify a distribution in $\mcd^\prime\big(\mathring{\bH}^{d+1}\times\mathring{\bH}^{d+1}\big)$ such that
\begin{equation*}
	\omega^{\bH}_{\alpha,2}(P_{\eta}f, f') = \omega^{\bH}_{\alpha,2}(f, P_{\eta}f') = 0 \, , \quad
	\omega^{\bH}_{\alpha,2}(f,f') - \omega^{\bH}_{\alpha,2}(f',f) = i G_{\bH,\alpha}(f,f') \, ,
\end{equation*}
for $f, \, f' \in C_0^{\infty}\big(\mathring{\bH}^{d+1}\big)$. Here, $G_{\bH,\alpha}$ is the causal propagator built in Proposition~\ref{prop:GHmodes}. Observe that, contrary to what happens for the algebra of free scalar fields on globally hyperbolic spacetimes, one has also to make sure that $\omega_{\alpha,2}$ is not only well-defined on $C^\infty_0(\mathring{\bH}^{d+1})$, but also on $\mathcal{C}_{\alpha,0}(\mathring{\bH}^{d+1})$ which is the space labelling all generators of $\mathcal{A}^{\rm on}_{\alpha}(\mathring{\bH}^{d+1})$. In \cite{Dappiaggi:2016fwc}, we have constructed via a mode expansion the ground state for a massive scalar field in $\PAdS_{d+1}$ or, equivalently, for a scalar field in $\mathring{\bH}^{d+1}$ obeying \eqref{eq:conformally_rescaled_dynamics}. In the following proposition we recollect the results already proven in \cite{Dappiaggi:2016fwc}, studying subsequently the properties of the ground state in the algebraic framework.

\begin{proposition} \label{prop:DNGAdS}
	The two-point function $\omega_{\alpha,2} \in \mathcal{D}^\prime(\PAdS_{d+1}\times\PAdS_{d+1})$ associated with the ground state for $\nu\in(0,1)$ and for $\alpha\in[\frac{\pi}{2},\pi]$ has integral kernel given by $\omega_{\alpha,2}(x,x') =  \lim_{\epsilon \to 0^+} \omega_{\alpha,2}(u_{\epsilon})$, with
	\begin{equation}\label{eq:state}
	\omega_{\alpha,2}(u_\epsilon) = \mathcal{N}_{\alpha} \left[ \cos(\alpha) \, \omega_2^{({\rm D})}(u_{\epsilon}) + \sin(\alpha) \, \omega_2^{({\rm N})}(u_{\epsilon}) \right] \, ,
	\end{equation}
	where $\mathcal{N}_{\alpha}$ is a normalization constant,
	\begin{equation*}
	\mathcal{N}_{\alpha} \doteq \frac{\Gamma \left(\frac{d-1}{2}\right)}{2^{d+1}\pi^{\frac{d+1}{2}} \left[\cos(\alpha)+\sin(\alpha)\right]} \, ,
	\end{equation*}

	\begin{subequations} \label{eq:Ghyp}
		\begin{align}
		\omega_2^{({\rm D})}(u_{\epsilon}) &= u_{\epsilon}^{-\frac{d}{2}-\nu} \, \frac{F \big(\tfrac{d}{2}+\nu, \tfrac{1}{2}+\nu; 1+2\nu; u_{\epsilon}^{-1}\big)}{\Gamma(1+2\nu)} \, , \label{eq:solution1} \\
		\omega_2^{({\rm N})}(u_{\epsilon}) &= u_{\epsilon}^{-\frac{d}{2}+\nu} \, \frac{F \big(\tfrac{d}{2}-\nu, \tfrac{1}{2}-\nu; 1-2\nu; u_{\epsilon}^{-1}\big)}{\Gamma(1-2\nu)} \, , \label{eq:solution2}
		\end{align}
	\end{subequations}
	$F$ is the Gaussian hypergeometric function and $u_{\epsilon} \doteq u(\sigma + 2i \epsilon (t-t') + \epsilon^2)$, with
	\begin{equation}
	u = u(\sigma) \doteq \cosh^2 \left(\frac{\sqrt{2\sigma}}{2}\right) \, .
	\end{equation}
	If $\nu\geq 1$, the only admissible option for the two-point function is \eqref{eq:state} with $\alpha=\pi$.
\end{proposition}

Observe that $G_{\alpha} \in \mathcal{D}^\prime(\PAdS_{d+1}\times\PAdS_{d+1})$ is nothing but the antisymmetric part of \eqref{eq:state}. We rewrite conveniently the two-point function:

\begin{proposition}
	The two-point function $\omega_{\alpha,2} \in \mathcal{D}^\prime(\PAdS_{d+1}\times\PAdS_{d+1})$ for $\nu \in (0,1) \setminus \{\frac{1}{2}\}$ has integral kernel given by $\omega_{\alpha,2}(x,x') =  \lim_{\epsilon \to 0^+} \omega_{\alpha,2}(u_{\epsilon})$, with
	
	\begin{equation} \label{eq:2pfsymmform}
	\omega_{\alpha,2}(u_\epsilon) = \mathcal{N}_{\alpha} \left[ A_{\alpha} \, \tilde{\omega}_2(u_{\epsilon}) + B_{\alpha} \, \iota_z \, \tilde{\omega}_2(u_{\epsilon}) \right] \, ,
	\end{equation}
	where
	\begin{align*}
	\tilde{\omega}_2(u_{\epsilon}) = F \big(\tfrac{d}{2}+\nu, \tfrac{d}{2}-\nu; \tfrac{d+1}{2}\nu; u_{\epsilon}\big) \, ,
	\end{align*}
	and 
	\begin{align*}
	A_{\alpha} &= \cos(\alpha) \Upsilon_A(\nu) + \sin(\alpha) \Upsilon_A(-\nu) \, , \\
	B_{\alpha} &= \cos(\alpha) \Upsilon_B(\nu) + \sin(\alpha) \Upsilon_B(-\nu) \, ,
	\end{align*}
	and
	\begin{align*}
	\Upsilon_A(\nu) \doteq \frac{(-1)^{\frac{d}{2}}}{2^{\frac{d-1}{2}}} 
	\frac{\Gamma \left(\frac{d}{2}+\nu\right) \Gamma \left(\frac{d}{2}-\nu\right) \Gamma \left(1+2\nu\right)}{\Gamma \left(\frac{d+1}{2}\right) \Gamma \left(\frac{d-1}{2}\right)} 
	= - (-1)^{\frac{d}{2}-\nu} \, \Upsilon_B(\nu) \, .
	\end{align*}
\end{proposition}

\begin{proof}
The proof is identical to the one of Lemma \ref{lem:convenient_rewriting}, hence we omit it.
\end{proof}

\begin{remark}
	Recall that the two-point functions $\omega_{\alpha,2} \in \mathcal{D}^\prime(\PAdS_{d+1}\times\PAdS_{d+1})$ and $\omega^{\bH}_{\alpha,2} \in \mathcal{D}^\prime\big(\mathring{\bH}^{d+1}\times\mathring{\bH}^{d+1}\big)$ are related by $\omega_{\alpha,2}(x,x') = (zz')^{\frac{d-1}{2}} \omega^{\bH}_{\alpha,2}(x,x')$. In particular, the analogue to \eqref{eq:2pfsymmform} may be written as
	\begin{equation}
	\omega_{\alpha,2}^{\bH}(u_\epsilon) = \mathcal{N}_{\alpha} \left[ A_{\alpha} \, \omega_2^{(+)}(u_{\epsilon}) + B_{\alpha} \, \iota_z \, \omega_2^{(-)}(u_{\epsilon}) \right]  \, ,
	\end{equation}
	with $\omega_2^{(\pm)}(u_{\epsilon}) \doteq (\pm zz')^{\frac{1-d}{2}} \tilde{\omega}_2(u_{\epsilon})$.
\end{remark}

We want to investigate whether the algebraic state whose two-point function is \eqref{eq:state} can be considered physically meaningful. The definition of a Hadamard state \cite{Radzikowski:1996ei,Radzikowski:1996pa} is valid for field theories in a globally hyperbolic spacetime and it cannot be straightforwardly generalized to $\mathring{\bH}^{d+1}$. In fact, as we saw in Section~\ref{sec:causalpropagator}, the causal propagator $G_{\bH,\alpha}$ has a richer singularity structure than that of one in a globally hyperbolic spacetime. In \cite{Dappiaggi:2016fwc} we have already investigated the local singular structure of \eqref{eq:state} in the case of $d=2,3$, comparing it with the one displayed in \cite{Dappiaggi:2014gea}. Here, we go one step beyond by looking for a microlocal characterization of the singularities of the two-point function. In particular, we show that there is full agreement with the results of \cite{Wrochna:2016ruq}. The reasons for pursuing this analysis are manifold, but one is especially relevant in our opinion. In globally hyperbolic spacetimes the definition of Hadamard state, as physically acceptable, yields that the wavefront set of its two-point function is the same as that of the Poincar\'e vacuum. This turns out to be not only of conceptual relevance, but also of practical use, especially for spacetime with no symmetries. We expect that, much in the same spirit, a physically acceptable state for a field theory in an asymptotically AdS spacetime will have the same singularities of the ground state in $\PAdS_{d+1}$. For this reason it is important to know the wavefront set of \eqref{eq:state}.

\begin{theorem}[Wavefront set of the two-point function]\label{th:WF}
	The wavefront set of the two-point function $\omega_{\alpha,2}$ coincides with that of $\omega_{\alpha,2}^\bH$, which in turn is given by
	\begin{equation} \label{WF}
	WF(\omega^{\bH}_{\alpha,2}) = \left\{ (x,k;x',k') \in T^*\big(\mathring{\bH}^{d+1} \times \mathring{\bH}^{d+1}\big) \setminus \{0\} : 
	(x,k) \sim_{\pm} (x',k'), \, k \triangleright 0 \right\} \, ,
	\end{equation}
	where $\sim_{\pm}$ means that there exist null geodesics with respect to the Minkowski metric $\gamma, \gamma^{(-)} : [0,1] \to\bR^{d+1}$ such that either
	$\gamma(0) = x = (\underline{x}, z)$ and $\gamma(1)=x'$ or $\gamma^{(-)}(0) = x^{(-)} = (\underline{x}, -z)$ and $\gamma^{(-)}(1) = x'$. Moreover,
	$k = (k_{\underline{x}}, k_z)$ ($k^{(-)} = (k_{\underline{x}}, -k_z)$) is coparallel to $\gamma$ ($\gamma^{(-)}$) at 0; and
	$-k'$ is the parallel transport of $k$ ($k^{(-)}$) along $\gamma$ ($\gamma^{(-)}$) at 1 (see Fig.~\ref{fig:wavefrontset}).
\end{theorem}

\begin{proof}
	First, we focus on the case $d=2,3$ for $\nu\in(0,1)$, recalling that $\alpha\in[\frac{\pi}{2},\pi]$. The case with $\nu\geq 1$ is equivalent to setting $\alpha=\pi$. It was shown in \cite{Dappiaggi:2016fwc} that the integral kernel of the two-point function $\omega_{\alpha,2} \in \mathcal{D}^\prime(\PAdS_{d+1}\times\PAdS_{d+1})$ reads
	$$\omega_{\alpha,2}(x,x') = H(x,x') + c_{\alpha,\nu} \, \iota_z H(x,x') + W_{\alpha}(x,x') \, , $$
	where $H$ is the Hadamard parametrix, $\iota_z$ the map defined in \eqref{inversion}, 
	\begin{equation} \label{eq:c_alphanu}
		c_{\alpha,\nu} \doteq i (-1)^{-\nu} \frac{\cos(\alpha) + (-1)^{-2\nu} \sin(\alpha)}{\cos(\alpha) + \sin(\alpha)} \, , 
	\end{equation}
	and $W_{\alpha}(x,x') \in C^{\infty}(\PAdS_{d+1}\times\PAdS_{d+1})$.
	From the analysis of Radzikowski \cite{Radzikowski:1996pa,Radzikowski:1996ei} and from the observation that $\iota_z$ implements the action of a discrete isometry of the Minkowski metric, the conclusion follows. For $d>3$, we start by observing from \eqref{eq:state} and from \eqref{eq:Ghyp} that the singular support of $\omega_{\alpha,2}$ contains those pair of points $x$, $x^\prime$ connected by a null geodesic either directly or after reflection at the boundary $z=0$. By adding the information that the state has been constructed by choosing positive frequencies, we can conclude that $WF(\omega_{\alpha,2})$, or more precisely, $WF(\omega_{\alpha,2}^{\bH})$ is contained or it coincides with the right hand side of \eqref{WF}. To infer that they are the same, one can follow the same argument of \cite[Lemma 6.4]{Fewster} together with \eqref{WF_G_H}.
\end{proof}

Hence, we propose the following definition for a Hadamard state for a scalar field in $\mathring{\bH}^{d+1}$. The definition for $\PAdS_{d+1}$ follows directly.

\begin{definition}[Hadamard state in $\mathring{\bH}^{d+1}$] \label{def:Hadamard_state}
	We call a state $\varpi^{\bH} : \mathcal{A}^{\rm on}_{\alpha}\big(\mathring{\bH}^{d+1}\big) \to \bC$ a quasifree, \emph{Hadamard state} for a scalar field in $\mathring{\bH}^{d+1}$ if its two-point function has a wavefront set of the form \eqref{WF}.
\end{definition}

In the next proposition we characterize the integral kernel of the two-point function of a quasifree, Hadamard state.\footnote{We are grateful to Klaus Fredenhagen for suggesting this line of reasoning.}

\begin{proposition} \label{prop:globaltolocal}
	Let $\varpi^{\bH} : \mathcal{A}^{\rm on}_{\alpha}\big(\mathring{\bH}^{d+1}\big) \to \bC$ a quasifree, Hadamard state for a scalar field in $\mathring{\bH}^{d+1}$. Then, the integral kernel of its two-point function reads
	\begin{equation} \label{eq:intker}
	\varpi^{\bH}_{\alpha,2}(x,x') = H^{\bH}(x,x') + c_{\alpha,\nu} \, \iota_z H^{\bH}(x,x') + W^{\bH}_{\alpha}(x,x') \, ,
	\end{equation}
	where $W^{\bH}_{\alpha}(x,x')\in C^\infty(\mathring{\bH}^{d+1}\times \mathring{\bH}^{d+1})$, $H^{\bH}(x,x')$ is the Hadamard parametrix for the operator $P_{\eta}$ in $\mathring{\bH}^{d+1}$ defined in \eqref{eq:conformally_rescaled_dynamics}, the map $\iota_z$ is defined in \eqref{inversion} and $c_{\alpha,\nu}$ in \eqref{eq:c_alphanu}. 
\end{proposition}

\begin{proof}
	Consider the ground state $\omega^{\bH}$ for a scalar field in $\mathring{\bH}^{d+1}$ satisfying \eqref{eq:conformally_rescaled_dynamics}. Write the two-point function associated to $\varpi^{\bH}$ as $\varpi_{2,\alpha}^{\bH} = \omega_{2,\alpha}^{\bH} + \Delta\omega_{2,\alpha}^{\bH}$, where $\Delta\omega_{2,\alpha}^{\bH} = \varpi_{2,\alpha}^{\bH} - \omega_{2,\alpha}^{\bH}$. By assumption, $WF(\varpi_{2,\alpha}^{\bH}) = WF(\omega_{2,\alpha}^{\bH})$ and since the antisymmetric part of $\varpi_{2,\alpha}^{\bH}$ and $\omega_{2,\alpha}^{\bH}$ coincide with the causal propagator, the integral kernel of $\Delta\omega_{2,\alpha}^{\bH}$ is symmetric. Combining this information with $WF(\Delta\omega_{2,\alpha}^{\bH}) \subseteq WF(\omega_{2,\alpha}^{\bH})$, it descends that $WF(\Delta\omega_{2,\alpha}^{\bH}) = \emptyset$. Hence, the integral kernel of $\varpi_{2,\alpha}^{\bH}$ differs from the one of $\omega_{2,\alpha}^{\bH}$ by a smooth term. The latter is of the sought form, as one can infer from the proof of Theorem~\ref{th:WF}.
\end{proof}

In the next proposition we show that this definition can be read as a generalization at the level of states of F-locality.

\begin{proposition} \label{prop:Flocality}
	Any quasifree, Hadamard state $\omega_{\alpha}^{\bH} : \mathcal{A}^{\rm on}_{\alpha}\big(\mathring{\bH}^{d+1}\big) \to \bC$ for a scalar field in $\mathring{\bH}^{d+1}$ obeying \eqref{eq:conformally_rescaled_dynamics}, in the sense of Definition~\ref{def:Hadamard_state}, is such that $\omega^{D}_{\alpha,2}$, the restriction to $D$ of the two-point function $\omega^{\bH}_{\alpha,2} \in \mcd^\prime\big(\mathring{\bH}^{d+1}\times\mathring{\bH}^{d+1}\big)$, has a wavefront set of Hadamard form
	\begin{align*}
	WF(\omega^{D}_{\alpha,2}) = \left\{ (x,k;x',k') \in T^*(D\times D) \setminus \{0\} : 
	(x,k) \sim (x',k'), \, k \triangleright 0 \right\}.
	\end{align*}
\end{proposition}

\begin{proof}
	Let $D$ be any globally hyperbolic region of $\mathring{\bH}^{d+1}$ and let $\omega^{D}_{\alpha,2}$ be as per hypothesis. In $D$ there cannot exist two points $x, \, x^\prime$ such that $x\sim_- x^\prime$ in the sense of Theorem~\ref{th:WF}. Hence, per direct inspection of \eqref{WF}, the wavefront set of $\omega^{D}_{\alpha,2}$ reduces to the one sought.
\end{proof}


\subsection{Wick ordering}
\label{sec:wickordering}

One of the key applications of Hadamard states in globally hyperbolic spacetimes is the construction of Wick polynomials. We investigate shortly this issue in the wake of Definition \ref{def:Hadamard_state} and Proposition~\ref{prop:globaltolocal}. Let $\varpi_{\alpha}^{\bH} : \mathcal{A}^{\rm on}_{\alpha}\big(\mathring{\bH}^{d+1}\big) \to \bC$ be a quasifree, Hadamard state, hence with integral kernel of the two-point function of the form \eqref{eq:intker}.

Next, consider the collection of regular functionals $\mathcal{F}_\alpha\big(\mathring{\bH}^{d+1}\big)$ and endow it with a new product $\star_{H,\alpha}:\mathcal{F}_\alpha\big(\mathring{\bH}^{d+1}\big)\times\mathcal{F}_\alpha\big(\mathring{\bH}^{d+1}\big)\to\mathcal{F}_\alpha\big(\mathring{\bH}^{d+1}\big)$ such that
\begin{equation}\label{deformed_product}
F_{\alpha}\star_{H,\alpha}F'_{\alpha} = \alpha_{H,\alpha}\left(\alpha_{H,\alpha}^{-1}(F_{\alpha})\star\alpha_{H,\alpha}^{-1}(F'_{\alpha}) \right) \, , \quad 
\forall F_{\alpha},F'_{\alpha} \in \mathcal{F}_\alpha\big(\mathring{\bH}^{d+1}\big) \, ,
\end{equation}
where $\star$ is the product \eqref{eq:algebraprodG}, while 
$$\alpha_{H,\alpha}\doteq\sum\limits_{n=0}^\infty\frac{\Gamma_{H,\alpha}^n}{n!}:\mathcal{F}_\alpha\big(\mathring{\bH}^{d+1}\big)\to\mathcal{F}_\alpha\big(\mathring{\bH}^{d+1}\big) $$
is defined in terms of 
$$\Gamma_{H,\alpha} \doteq -i\int_{\mathring{\bH}^{d+1}\times\mathring{\bH}^{d+1}} \left[H^{\bH}(x,x') + c_{\alpha,\nu} \, \iota_z H^{\bH}(x,x')\right] \frac{\delta}{\delta\phi(x)}\otimes\frac{\delta}{\delta \phi(x')} \, . $$
The algebra $\big(\mathcal{F}_{\alpha} \big(\mathring{\bH}^{d+1}\big), \star_{H,\alpha} \big)$ is isomorphic to $\mathcal{A}_{\alpha}\big(\mathring{\bH}^{d+1}\big) = \big(\mathcal{F}_{\alpha} \big(\mathring{\bH}^{d+1}\big), \star_\alpha \big)$, but  the new product allows for the inclusion of local non linear functionals, the so-called microcausal functionals. We recall here the definition adapted to our context:

\begin{definition}
	We call $F_{\alpha} : \mathcal{C}_{\alpha}\big(\mathring{\bH}^{d+1}\big) \to \bC$ a \emph{microcausal functional} if, for all $n \geq 1$ and for all $\Phi_{\alpha} \in \mathcal{C}_{\alpha}\big(\mathring{\bH}^{d+1}\big)$, $F_{\alpha}^{(n)}(\Phi_{\alpha}) \in \mathcal{E}'\big(\mathring{\bH}^{d+1}\big)^{\otimes n}$. Only a finite number of functional derivatives do not vanish and $WF(F_{\alpha}^{(n)}) \subset \Xi_n$, with
	\begin{equation*}
		\Xi_n \doteq T^*\big(\mathring{\bH}^{d+1}\big)^n \setminus \left\{ (x_1, \ldots, x_n, k_1, \ldots, k_n) | (k_1, \ldots, k_n) \in \left.\big(\bar{V}_+^n \cup \bar{V}_-^n\big)\right|_{(x_1, \ldots, x_n)} \right\}
	\end{equation*}
	where $\bar{V}_\pm$ are the subsets of $T^*\big(\mathring{\bH}^{d+1}\big)$ formed by elements $(x_i,k_i)$ in which each $k_i$ lies in the closed future $(+)$ and in the closed past $(-)$ light cone. The set of microcausal functionals is denoted by $\mathcal{F}_{\mu,\alpha}\big(\mathring{\bH}^{d+1}\big)$.
\end{definition}

\begin{definition}\label{Wick-polynomials}
	We call $\mathcal{A}_{\mu,\alpha} \big(\mathring{\bH}^{d+1}\big) \equiv \big(\mathcal{F}_{\mu}\big(\mathring{\bH}^{d+1}\big), \star_{H,\alpha} \big)$ the \emph{extended algebra of Wick polynomials} in $\mathring{\bH}^{d+1}$.
\end{definition}

\begin{remark}
	In view of Definition \ref{Wick-polynomials}, one might wonder whether, choosing any globally hyperbolic subregion $D\subset\mathring{\bH}^{d+1}$, the algebra of Wick polynomials in $\mathring{\bH}^{d+1}$ restricted to $D$, $\mathcal{A}_{\mu,\alpha}(D)$, agrees with $\mathcal{A}_{\mu}(D)$. This is the collection of microcausal functionals supported in $D$ with the product $\star_{H}$. This is defined as $\star_{H,\alpha}$ though with $\Gamma_{H,\alpha}$ replaced by 
	$$\Gamma_H = -i\int_{\mathring{\bH}^{d+1}\times\mathring{\bH}^{d+1}}  H^{\bH}(x,x') \frac{\delta}{\delta\phi(x)}\otimes\frac{\delta}{\delta \phi(x')} \, . $$
	By direct inspection one can realize that the two algebra differ, since in $\Gamma_{H,\alpha}$ there is the additional contribution $c_{\alpha,\nu} \, \iota_z H^{\bH}(x,x')$. Yet, in view of Proposition \ref{prop:Flocality}, this new term is smooth and thus the algebra in $D$ stemming from Definition \ref{Wick-polynomials} and the one constructed only using $H^{\bH}(x,x')$ are $*$-isomorphic \cite{Hollands:2001nf}.  
	
	The reason for sticking to Definition \ref{Wick-polynomials} and not to the standard one is due to the fact that the local definition of Wick polynomials fails to intercept the singularities due to the light rays reflected at the boundary $\partial\bH^{d+1}$. In other words, one cannot start from each $\mathcal{A}_\mu(D)$ to reconstruct a global algebra of Wick polynomials. Nonetheless, it is worth mentioning that, in some specific scenarios, it might be physically meaningful to consider the algebra of Wick polynomials constructed locally. A special instance was investigated in \cite{Dappiaggi:2014gea}.
\end{remark}


\section{Conclusions}

We have analysed the algebraic quantization of a real, massive scalar field in $\PAdS_{d+1}$ in terms of an equivalent theory in $\mathring{\bH}^{d+1}$. Although it is customary to summarize the results of the paper, we prefer instead only to highlight two directions for future investigations which we deem noteworthy. 
The first concerns the generalization of our construction to stationary spacetimes with a timelike boundary, the second to the so-called asymptotically AdS spacetimes. Concerning the latter, there are many inequivalent definitions in the literature and it is unclear to us which is the best possible choice. In addition, all our results, ranging from the study of admissible boundary conditions and of the associated causal propagator to the existence of a suitable notion of Hadamard states need to be investigated from scratch. Especially if one has the long term goal of further applying the algebraic approach to the realm of the AdS/CFT correspondence and  more in general of the holographic principle, this is in our opinion the first question to answer.


\section*{Acknowledgments}  
We are grateful to Nicol\`o Drago, Klaus Fredenhagen, Igor Khavekine, Jorma Louko, Gabriele Nosari, Pedro Lauridsen Ribeiro and Nicola Pinamonti for useful comments and discussions, and especially to Micha{\l} Wrochna for enlightening remarks concerning Section 4. The work of C.D. was supported by the University of Pavia. The work of H.~F. was supported by the INFN postdoctoral fellowship ``Geometrical Methods in Quantum Field Theories and Applications''.


\end{document}